%% file: main.tex
\documentclass[a4paper,UKenglish,cleveref,autoref,thm-restate]{lipics-v2021}

\allowdisplaybreaks

\hideLIPIcs  

\title{Strong Normalisation for Asynchronous Effects}

\author{Danel Ahman}{Institute of Computer Science, University of Tartu, Estonia}{danel.ahman@ut.ee}{https://orcid.org/0000-0001-6595-2756}{}

\author{Ilja Sobolev}{Institute of Mathematics and Statistics, University of Tartu, Estonia}{ilja.sobolev@ut.ee}{https://orcid.org/0009-0000-0939-6388}{} 

\authorrunning{D. Ahman and I. Sobolev}

\Copyright{Danel Ahman and Ilja Sobolev}



\category{} 

\supplementdetails[subcategory={Formalisation}]{Software}{https://doi.org/10.5281/zenodo.19695414}

\funding{This work was supported by the Estonian Research Council grant PRG2764.}

\nolinenumbers 

\EventEditors{Frank Pfenning}
\EventNoEds{1}
\EventLongTitle{11th International Conference on Formal Structures for Computation and Deduction (FSCD 2026)}
\EventShortTitle{FSCD 2026}
\EventAcronym{FSCD}
\EventYear{2026}
\EventDate{July 20--23, 2026}
\EventLocation{Lisbon, Portugal}
\EventLogo{}
\SeriesVolume{378}
\ArticleNo{11}

\input{macros.tex}

\begin{document}

\maketitle

\begin{abstract}
  Asynchronous effects of Ahman and Pretnar complement the conventional
  synchronous treatment of algebraic effects with asynchrony based on decoupling
  the execution of algebraic operation calls into signalling that an operation's
  implementation needs to be executed, and into interrupting a running
  computation with the operation's result, to which the computation can react by
  installing matching interrupt handlers. 
  Beyond providing asynchrony for algebraic effects, the resulting core calculus
  also naturally models examples such as pre-emptive multi-threading,
  (cancellable) remote function calls, and multi-party applications.
  In this paper, we study the normalisation properties of this calculus.
  We prove that if one removes general recursion from it, then the remaining
  calculus is strongly normalising, including both its sequential and parallel
  parts.
  To cover more interesting programs, we also prove that the sequential part of the
  calculus remains strongly normalising when a controlled amount of
  interrupt-driven recursive behaviour is reintroduced. Our normalisation proofs
  are structured compositionally as an extension of Lindley and Stark's
  $\top\top$-lifting-based approach for proving strong normalisation of
  effectful languages. All our results are also formalised in Agda.
\end{abstract}

\input{introduction}
\input{asynchronous-effects}
\input{normalisation-sequential}
\input{normalisation-reinstallable}
\input{normalisation-parallel}

\input{conclusion}

\bibliography{references}

\appendix
\input{appendix}

\end{document}

%% file: macros.tex
\usepackage{bbold} 
\usepackage{graphicx}
\usepackage{lineno}
\usepackage{listings}
\usepackage{stmaryrd}
\usepackage{thmtools}
\usepackage{upquote}
\usepackage{xcolor}
\usepackage[all]{xy}
\usepackage[most]{tcolorbox}

\usepackage{mathpartir} 

\usepackage{amssymb} 
\usepackage{framed} 
\usepackage{mathtools} 
\usepackage{semantic} 

\usepackage{todonotes} 

\usepackage{float} 

\usepackage{amsmath}

\definecolor{codegreen}{rgb}{0,0.6,0}
\definecolor{codegray}{rgb}{0.5,0.5,0.5}
\definecolor{codepurple}{rgb}{0.58,0,0.82}
\definecolor{backcolour}{rgb}{0.95,0.95,0.92}

\definecolor{keywordColor}{rgb}{0.0,0.0,0.5} 
\definecolor{rulenameColor}{rgb}{0.5,0.5,0.5} 

\newcommand{\defeq}{\mathrel{\overset{\text{\tiny def}}{=}}} 

\newcommand{\lambdaAEff}{$\lambda_{\text{\ae}}$} 

\newcommand{\bnfis}{\mathrel{\;{:}{:}{=}\ }}
\newcommand{\bnfor}{\mathrel{\,\;\big|\ \ \!}}

\newcommand{\opsym}[1]{\mathsf{#1}} 
\newcommand{\op}{\opsym{op}} 

\newcommand{\sig}{\Sigma} 

\renewcommand{\o}{o} 
\renewcommand{\i}{\iota} 

\newcommand{\opincomp}[2]{{\mathsf{#1}}\,{\tmkw{\downarrow}}\,#2} 

\newcommand{\at}{\mathbin{!}} 
\newcommand{\att}{\mathbin{!!}} 

\newcommand{\tysym}[1]{\mathsf{#1}}
\newcommand{\tybase}{\tysym{b}} 
\newcommand{\tyunit}{\mathsf{unit}} 
\newcommand{\typrod}[2]{#1 \times #2} 
\newcommand{\tysum}[2]{#1 + #2} 
\newcommand{\tyfun}[2]{#1 \to #2} 
\newcommand{\typromise}[1]{\langle #1 \rangle} 

\newcommand{\tycomp}[2]{#1 \at #2} 

\newcommand{\tyrun}[3]{#1 \att (#2,#3)} 
\newcommand{\typar}[2]{#1 \mathbin{\tmkw{\vert\vert}}  #2} 
\newcommand{\tyC}{C} 
\newcommand{\tyD}{D} 

\newcommand{\tm}[1]{\mathsf{#1}} 
\newcommand{\tmkw}[1]{\tm{\color{keywordColor}#1}} 

\newcommand{\tmpromise}[1]{\langle #1 \rangle} 

\newcommand{\tmunit}{()} 
\newcommand{\tmpair}[2]{( #1 , #2 )} 
\newcommand{\tminl}[2][]{\tmkw{inl}_{#1}\,#2} 
\newcommand{\tminr}[2][]{\tmkw{inr}_{#1}\,#2} 
\newcommand{\tmfun}[2]{{\mathop{\tmkw{fun}}}\; (#1) \mapsto #2} 
\newcommand{\tmapp}[2]{#1\,#2} 

\newcommand{\tmreturn}[2][]{\tmkw{return}_{#1}\, #2} 
\newcommand{\tmlet}[3]{\tmkw{let}\; #1 = #2 \;\tmkw{in}\; #3} 

\newcommand{\tmop}[4]{\tm{#1}\;(#2, #3. #4)} 
\newcommand{\tmopin}[3]{\tmkw{\downarrow}\, \tm{#1}\,(#2, #3)} 
\newcommand{\tmopout}[3]{\tmkw{\uparrow}\,\tm{#1}\, (#2, #3)} 
\newcommand{\tmopoutbig}[3]{\tmkw{\uparrow}\,\tm{#1}\, \big(#2, #3\big)} 

\newcommand{\tmmatch}[3][]{\tmkw{match}\;#2\;\tmkw{with}\;\{#3\}_{#1}} 

\newcommand{\tmawait}[3]{\tmkw{await}\;#1\;\tmkw{until}\;\tmpromise{#2}\;\tmkw{in}\;#3} 

\newcommand{\tmwith}[5]{\tmkw{promise}\; (\tm{#1}\; #2 \mapsto #3)\; \tmkw{as}\; #4\; \tmkw{in}\; #5} 
\newcommand{\tmwithbig}[5]{\tmkw{promise}\; \big(\tm{#1}\; #2 \mapsto #3\big)\; \tmkw{as}\; #4\; \tmkw{in}\; #5} 
\newcommand{\tmwithre}[6]{\tmkw{promise}\; (\tm{#1}\; #2 \, #3 \mapsto #4)\; \tmkw{as}\; #5\; \tmkw{in}\; #6} 

\newcommand{\tmrun}[1]{\tmkw{run}\; #1} 
\newcommand{\tmpar}[2]{#1 \mathbin{\tmkw{\vert\vert}} #2} 

\newcommand{\reduces}{\leadsto} 

\newcommand{\E}{\mathcal{E}} 
\newcommand{\F}{\mathcal{F}} 

\newcommand{\types}{\vdash} 
\newcommand{\of}{\mathinner{:}} 

\newcommand{\coopinfer}[3]{\inferrule*[Lab={\color{rulenameColor}#1}]{#2}{#3}}

\newcommand{\order}[1]{\sqsubseteq_{#1}} 

\makeatletter
\providecommand*{\cupdot}{
  \mathbin{%
    \mathpalette\@cupdot{}%
  }%
}
\newcommand*{\@cupdot}[2]{%
  \ooalign{%
    $\m@th#1\cup$\cr
    \hidewidth$\m@th#1\cdot$\hidewidth
  }%
}
\makeatother

\definecolor{redexColor}{rgb}{0.83, 0.83, 0.83} 

\newcommand{\VRed}{\mathbf{VRed}}
\newcommand{\CRed}{\mathbf{CRed}}
\newcommand{\KRed}{\mathbf{KRed}}

\newcommand{\SRed}{\mathbf{SRed}}
\newcommand{\ARed}{\mathbf{ARed}}

\newcommand{\capp}{\mathbin{@}}
\newcommand{\Id}{\textit{Id}}

\newcommand{\copin}[2]{\tmkw{\downarrow}\, \tm{#1}\,(#2)}

\newcommand{\Ren}[2]{\mathsf{Ren}\;#1\;#2}

\newcommand{\SN}{\mathcal{SN}}

\newcommand{\copout}[2]{\tmkw{\uparrow}\, \tm{#1}\,(#2)}

\newcommand{\xleadsto}[2][]{\ext@arrow 0359\leadstofill@{#1}{#2}}

\newcommand{\sizei}[1]{\lvert #1 \rvert_i}
\newcommand{\tmin}{\tmkw{\downarrow}\,}
\newcommand{\tmout}{\tmkw{\uparrow}\,}

%% file: introduction.tex
\section{Introduction}
\label{sec:intro}

It is interesting, and often even crucial, to know when certain classes of
programs are guaranteed to terminate and not enter an infinite loop.
Even in inherently non-terminating programming languages, it is
useful to know that certain fragments are strongly
normalising, e.g., to guarantee that user queries or sensor readings are
processed in finite time.
In this paper, we study the normalisation properties of a core
$\lambda$-calculus proposed by Ahman and Pretnar~\cite{Ahman:POPL,Ahman:LMCS},
called~\lambdaAEff, which models asynchronous programming with algebraic
effects. This calculus was originally designed to address the
synchrony in the operational treatment of algebraic
effects~\cite{Bauer:AlgebraicEffects,Kammar:Handlers,Plotkin:HandlingEffects},
in which a program whose execution encounters an operation call $\tmop {op} V y
M$ is blocked until some implementation of $\op$ finishes executing and unblocks
$M$.
Importantly, however, Ahman and Pretnar also showed that \lambdaAEff~naturally models a number of
other examples, such as pre-emptive multi-threading, (cancellable) remote
function calls, multi-party applications, and even a parallel
variant of runners of algebraic effects~\cite{Ahman:Runners}.

The \lambdaAEff-calculus of Ahman and Pretnar addresses the forced synchrony by
decoupling the different phases involved in executing operation calls into
separate programming abstractions, giving the programmer fine-grained control over
when the execution should block, and when blocking is not needed. First,
programmers can issue \emph{outgoing signals} $\tmopout{op}{V}{M_{signal}}$,
with which to indicate that some implementation of $\op$ needs to be executed.
Next, from the perspective of $\op$'s implementations, these signals become
\emph{incoming interrupts} $\tmopin{op}{V}{M_{interrupt}}$, to which these
implementations can react by installing \emph{interrupt handlers}
$\tmwith{op}{x}{M_\op}{p}{N_{handler}}$, where the handler code $M_\op$ is
executed only when a corresponding interrupt is intercepted by the interrupt
handler. Finally, the variable $p$ bound in the continuation $N_{handler}$ of an
interrupt handler carries a distinguished \emph{promise type} $\typromise{X}$,
which allows its value to be explicitly \emph{awaited} with
$\tmawait{p}{x}{N_{await}}$, blocking the execution of the rest of the program
$N_{await}$ until the promise is fulfilled with a value. Asynchrony is achieved
because the continuations $M_{signal}$, $M_{interrupt}$, and $N_{handler}$ of
signals, interrupts, and interrupt handlers have non-blocking semantics,
while the continuation $N_{await}$ of awaiting has blocking semantics. 
In order to model both an operation's call site and its
implementation, the \lambdaAEff-calculus organises individual sequential
computations into parallel processes that communicate by issuing signals and
reacting to received interrupts. 

The advantage of the \lambdaAEff-calculus is that the same programming
abstractions can be used both (i) for expressing algebraic operation calls
$\tmop{op}{V}{y}{M}$, by first issuing a signal to request the execution of
$\op$'s implementation and then waiting for an interrupt with $\op$'s result to
arrive and be handled, and (ii) for defining the implementation(s) of $\op$, by
first waiting for an interrupt with the request to execute $\op$'s
implementation and then issuing a signal with the implementation's result back
to the operation's call site. Importantly, however, the call site does not have
to block immediately to await the operation's result, but it can meanwhile
asynchronously perform other tasks.
For instance, the following program
\[
  \begin{array}{@{}l}
    \tmopout{encrypt\text{-}request}{\tmpair{plaintext}{key}}{\tmreturn{\tmunit}};
    \\[1pt]
    \tmwith{encrypt\text{-}response}{ciphertext}{\tmreturn[]{\tmpromise{ciphertext}}}{promised\text{-}ciphertext}{
    \\[1pt]
    \ldots
    \\[1pt]
    \tmawait{promised\text{-}ciphertext}{ciphertext}{M}}
  \end{array}
\]
illustrates how the call site of some (costly) encryption operation $\tm{encrypt} :
\typrod{\tm{string}}{\tm{string}} \to \tm{string}$ can be modelled by splitting
it into issuing an $\tm{encrypt\text{-}request} :
\typrod{\tm{string}}{\tm{string}}$ signal and installing an interrupt handler to
await an $\tm{encrypt\text{-}response} : \tm{string}$ interrupt, after which
the computation can perform other tasks (here denoted by $\ldots$) without
blocking. The computation only blocks when it reaches the $\tmkw{await}$
whose continuation depends on the ciphertext promised by $\tm{encrypt}$. We
can then implement $\tm{encrypt}$ in another process, e.g., as
\[
  \begin{array}{@{}l}
    \tmwithbig{encrypt\text{-}request}{\tmpair{plaintext}{key}}{
    \\[1pt]
    \quad\tmlet{ciphertext}{perform\text{-}encryption\tmpair{plaintext}{key}}{
    \\[1pt]
    \quad\tmopout{encrypt\text{-}response}{ciphertext}{\tmreturn[]{\tmpromise{\tmunit}}}}
    \\[1pt]
    }{p}{\tmreturn[]{p}}
  \end{array}
\]
by dually first awaiting an $\tm{encrypt\text{-}request} :
\typrod{\tm{string}}{\tm{string}}$ interrupt to arrive, and once it arrives and
the ciphertext is computed, eventually issuing the $\tm{encrypt\text{-}response}
: \tm{string}$ signal to communicate the ciphertext back to the $\tm{encrypt}$
operation's call site. We note that generalisations of interrupt handlers
discussed in \cref{sec:sn-reinstallable} also allow the interrupt handler for
$\tm{encrypt\text{-}request}$ to be reinstalled, so as to be able to process
further encryption requests.

The calculus is, however, flexible enough that not all signals have to have a
corresponding interrupt, and not all interrupts have to be responses to some
signals, allowing programmers to naturally model not only asynchronous algebraic
operation calls, but also server-like and remote code execution examples,
and moreover spontaneous behaviour, such as a user clicking a
button, a sensor sending data, or the environment pre-empting a thread.

The goal of this paper is to study the normalisation properties of the
\lambdaAEff-calculus (reviewed in \cref{sec:aeff}), to give termination
guarantees for asynchronous programs written in or modelled using it.
\emph{First}, we prove that when one removes general recursion from the original
formulation of \lambdaAEff, its sequential fragment, which models the code being
executed by individual asynchronous computations, becomes strongly normalising
(\cref{sec:sn-seq}). \emph{Second}, we show that if one reintroduces a
controlled amount of recursive behaviour into the programs via (a variant of)
the concept of reinstallable interrupt handlers (originally proposed by Ahman
and Pretnar to be able to model interesting asynchronous examples in the absence
of general recursion), then the sequential part of \lambdaAEff~remains strongly
normalising (\cref{sec:sn-reinstallable}). \emph{Third}, we establish that
without adding reinstallable interrupt handlers, the parallel part of the
general-recursion-free \lambdaAEff-calculus is also strongly normalising,
confirming the intuition that reinstallable interrupt handlers indeed add
expressive and computational power (\cref{sec:sn-parallel}).

As \lambdaAEff~is higher-order and effectful, our proofs are naturally
structured as an extension of Lindley and Stark's $\top\top$-lifting-based
approach~\cite{Lindley:TopTopLifting}, which itself is an extension of the
Girard-Tait-style type-directed reducibility approach to proving strong
normalisation of higher-order languages~\cite{proofs,tait}. We demonstrate that
the reducibility and $\top\top$-lifting-based approaches naturally extend to the
\emph{challenging features} of \lambdaAEff: it is effectful, it contains
sequential and parallel parts, it is highly non-deterministic and even
non-confluent~\cite{Ahman:POPL,Ahman:LMCS}, it includes programming abstractions
based both on algebraic
effects~\cite{Plotkin:NotionsOfComputation,Plotkin:GenericEffects} (signals,
interrupt handlers, awaiting) and effect
handling~\cite{Bauer:AlgebraicEffects,Kammar:Handlers,Plotkin:HandlingEffects}
(interrupts), it includes scoped algebraic
operations~\cite{Pirog:ScopedOperations} (interrupt handlers), it includes
reduction rules corresponding to the commutativity of algebraic operations
(signals and interrupt handlers), and it includes computations with
non-blocking continuations that bind variables (interrupt handlers).
All the work presented in this paper is also formalised in
Agda. The formalisation is available via Zenodo \cite{sobolev:formalisation}.

This paper is based on and extends the Bachelor's thesis~\cite{sobolev:bsc} of
one of the authors.

%% file: asynchronous-effects.tex
\section{\lambdaAEff: A Core Calculus for Asynchronous Algebraic Effects}
\label{sec:aeff}

We now recall the syntax and small-step operational semantics of the
\lambdaAEff-calculus without general recursion, as presented in the journal
version of Ahman and Pretnar's work~\cite{Ahman:LMCS}.

\subsection{Sequential Part}
\label{sec:aeff:seq}

The sequential part of the \lambdaAEff-calculus is based on Levy et al.'s
fine-grain call-by-value $\lambda$-calculus (FGCBV)~\cite{Levy:FGCBV}, in which \emph{terms}
are stratified into two classes, values and computations, respectively ranged over by
$V,W,\ldots$ and $M,N,\ldots$. \emph{Values} are defined by the grammar
\begin{align*}
  V, W
  \bnfis& x                                       & &\text{variable} \\
  \bnfor& \tmfun{x : X}{M}                        & &\text{function abstraction} \\
  \bnfor& \tmpromise V                            & &\text{fulfilled promise}
  \\[1ex]
  \intertext{and (sequential) \emph{computations} are defined by the grammar}
  M, N
  \bnfis& \tmreturn{V}                            & &\text{returning a value} \\
  \bnfor& \tmlet{x}{M}{N}          & &\text{sequencing} \\
  \bnfor& V\,W                                    & &\text{function application} \\
  \bnfor& \tmopout{op}{V}{M}       & &\text{outgoing signal} \\
  \bnfor& \tmopin{op}{V}{M}          & &\text{incoming interrupt} \\
  \bnfor& \tmwith{op}{x}{M}{p}{N}      & &\text{interrupt handler} \\
  \bnfor& \tmawait{V}{x}{M}             & &\text{awaiting a promise to be fulfilled}
\end{align*}
We omit terms corresponding to unit, product, and sum types as they can be
accommodated in a standard way, both in syntax and semantics, and in the
normalisation proofs; see~\cref{sec:sn-seq:sums-products}.

As is typical in FGCBV-based languages, \emph{values} include variables (drawn
from a countable set, and ranged over by $x, y, z, \ldots$) and function
abstractions, whose bodies are possibly effectful computations.
The bound variable $x$ denoting the function's parameter is annotated with its
type $X$ (defined later in this section). 
The only new value term compared to FGCBV is $\tmpromise V$, which is used to
fulfil a promise with the value $V$ to unblock blocked computations.

\emph{Computations} similarly contain standard FGCBV term-formers,
which are extended with \lambdaAEff-specific constructs. The standard
computations include returning values with $\tmreturn[]{V}$ and sequencing
computations with $\tmlet{x}{M}{N}$, and applying a function $V$ to an argument
$W$, written $V\,W$. The \lambdaAEff-specific computations include all the
programming abstractions discussed in \cref{sec:intro}: outward-propagating
signals $\tmopout{op}{V}{M}$, inward-propagating interrupts
$\tmopin{op}{V}{M}$, installed interrupt handlers $\tmwith{op}{x}{M}{p}{N}$, and
the explicit blocking until a given promise is fulfilled with
$\tmawait{V}{x}{M}$. The names $\op$ used in signals and interrupts are drawn
from an assumed finite set $\Sigma$ of operation names.

Next, \emph{well-typed values and computations} are classified by two typing
judgements: ${\Gamma \types V : X}$ and ${\Gamma \types M :
\tycomp{X}{(\o,\i)}}$. In these typing judgements, $\Gamma$ is a context of
variables of the form $x_1 \of X_1, \ldots, x_n \of X_n$, $X$ is a type of
values (returned by computations), and $(\o,\i)$ is an effect annotation
specifying which signals a computation might issue and which interrupt handlers
it might have installed (we discuss $\o$ and $\i$ in more detail later in this
section). 

The grammar of \emph{value types}
(ranged over by $X, Y, \ldots$) is given by
\[
X,Y \bnfis \tybase \bnfor \tyfun{X}{\tycomp{Y}{(\o,\i)}} \bnfor \typromise{X}
\] 
and \emph{computation types} are written as $\tycomp{X}{(\o,\i)}$. Here $\tybase$ ranges over base
types. We additionally assume that every operation $\op \in \Sigma$ is assigned
a \emph{typing signature}, written $\op \of A_\op$, where for the time being we simply let $A_\op
::= \tybase$, but note that more generally this set of \emph{ground types}
can also contain finite sums and products of base types. 
Ahman and Pretnar also showed how modal types can be used to further allow
function types among such ground types~\cite{Ahman:LMCS}.

The effect annotations $\o$ and $\i$ are drawn from sets $O$ and $I$, where $O
\defeq \mathcal{P}(\Sigma)$ and $I$ is defined as a fixed point of the (set)
functor $\Phi : \mathsf{Set} \to \mathsf{Set}$, which is given on objects by
$\Phi (Z) \defeq \Sigma \Rightarrow (O \times Z)_\bot$, where $\Rightarrow$
denotes a set of functions and $(-)_\bot$ is the lifting operation given by
disjoint union $(-) \cupdot \{\bot\}$. 
Intuitively, with $I$ defined as a fixed point of $\Phi$, each $\i \in I$ is a
nesting of partial mappings of names in $\Sigma$ to pairs of $O$- and
$I$-annotations---these classify the possible effects of the interrupt handler
code associated with a given interrupt name by some interrupt handler.
For defining the \lambdaAEff-calculus, both the least fixed point (an inductive
definition) and the greatest fixed point (a coinductive definition) of $\Phi$
work.
The least fixed point of $\Phi$ gives us annotations $\i$ with finite depth, and
the greatest fixed point gives us annotations $\i$ with possibly infinite depth.
While the latter is needed to assign types to recursive examples (see \cref{sec:sn-reinstallable}), we demonstrate
in \cref{sec:sn-parallel} that the former is useful for proving strong
normalisation of parallel processes in the absence of recursive features.

We write $\i\, (\op_i) = (\o_i,\i_i)$ to mean that the annotation $\i$ maps
$\op_i$ to $(\o_i,\i_i)$. These sets of effect annotations naturally carry
partial orders, with $\order{O}$ given by subset inclusion and $\order{I}$
defined (co)recursively pointwise. These partial orders in turn induce
component-wise a partial order on the product set $O \times I$, written $\order{O
\times I}$ and used in the sub-effecting rule. 

The effect annotations also support an action of operations
\[
\opincomp {op} {(\o , \i)}
~\defeq~
  \begin{cases}
   \left(\o \sqcup \o' , \i[\op \mapsto \bot] \sqcup \i' \right) & \mbox{if } \i\, (\op) = (\o',\i')\\
   \left(\o,\i\right) & \mbox{otherwise}
  \end{cases}
\]
that mimics the triggering of interrupt handlers by matching interrupts at the
effect-typing level, and is used for typing interrupts below. For more details
on $O$ and $I$, see~\cite{Ahman:POPL,Ahman:LMCS}.

For the \emph{typing judgements} $\Gamma \types V : X$ and $\Gamma \types M :
\tycomp{X}{(\o,\i)}$, we only present the rules for \lambdaAEff-specific
computations to illustrate them, and refer the reader
to~\cite{Ahman:POPL,Ahman:LMCS} for other rules: 
\begin{mathpar}
  \coopinfer{}{
    \op \in \o \\
    \Gamma \types V : A_\op \\
    \Gamma \types M : \tycomp{X}{(\o,\i)}
  }{
    \Gamma \types \tmopout{op}{V}{M} : \tycomp{X}{(\o,\i)}
  }
  \qquad
  \coopinfer{}{
    \Gamma \types V : A_\op \\
    \Gamma \types M : \tycomp{X}{(\o,\i)}
  }{
    \Gamma \types \tmopin{op}{V}{M} : \tycomp{X}{\opincomp {op} (\o,\i)}
  }
  \\
  \coopinfer{}{
    ({\o'} , {\i'}) = \i\, (\op)  \\
    \Gamma, x \of A_\op \types M_\op : \tycomp{\typromise X}{(\o',\i')} \\
    \Gamma, p \of \typromise X \types N : \tycomp{Y}{(\o,\i)}
  }{
    \Gamma \types \tmwith{op}{x}{M_\op}{p}{N} : \tycomp{Y}{(\o,\i)}
  }
  \\
  \coopinfer{}{
    \Gamma \types V : \typromise X \\
    \Gamma, x \of X \types M : \tycomp{Y}{(\o,\i)}
  }{
    \Gamma \types \tmawait{V}{x}{M} : \tycomp{Y}{(\o,\i)}
  }
  \qquad
   \coopinfer{}{
      \Gamma \types M : \tycomp{X}{(\o, \i)} \\
      (\o,\i) \order {O \times I} (\o',\i')
    }{
      \Gamma \types M : \tycomp{X}{(\o', \i')}
    }
\end{mathpar}
Note how the interrupt handler rule requires the handler code $M_\op$ to be
typed at an effect annotation determined by $\i$ at $\op$. Also observe how
$\opincomp {op} {(\o , \i)}$ is used on the effect annotations of interrupts
$\tmopin{op}{V}{M}$ to express that it triggers interrupt handlers for $\op$ in
$M$.

We conclude the sequential part of \lambdaAEff~by recalling its small-step
operational semantics. It is given by a \emph{reduction relation} $M \reduces N$
defined in \cref{fig:seq-reduction-rules}. The relation consists of standard
FGCBV computation rules \eqref{r1}--\eqref{r2}, expressing how function
applications and sequencing work. The \lambdaAEff-specific rules
\eqref{r3}--\eqref{r5} express that signals, interrupts, and awaiting behave
semantically like algebraic operations, by structurally propagating outwards in
computations, which expresses that if they occur in the first sequenced
computation, then they also occur first in the overall computation.
The rule \eqref{r6} expresses that signals propagate outwards through interrupt
handlers, so as to reach any other parallel processes expecting them. 
In this and other similar rules, which do not explicitly require that $\tm{op}
\neq \tm{op'}$, the names $\tm{op}$ and $\tm{op'}$ may coincide---this is because we
assume a single global set $\Sigma$ of signal and interrupt names. 
The rules \eqref{r7}--\eqref{r11} express the effect-handling-style behaviour of
interrupts, in that they recursively traverse the given computation, where they
trigger interrupt handlers for matching interrupts (rule \eqref{r9}) and move
past interrupt handlers for non-matching interrupts (rule \eqref{r10}).
Interrupts get discarded when they reach a return value and there are no more
interrupt handlers to be triggered (rule \eqref{r7}), and they propagate past
outward-propagating signals (rule \eqref{r8}) and into the continuations of
blocking computations (rule \eqref{r11}). The latter is not necessary
\textit{per se}, but it results in a simpler meta-theory~\cite{Ahman:LMCS}. The
rule \eqref{r12} specifies how fulfilled promises unblock awaiting computations.
Finally, the rule \eqref{r13} uses evaluation contexts to allow execution to
take place in subcomputations. Observe that $\tmkw{await}$ is not included in
evaluation contexts, modelling the blocking behaviour of this construct.

\begin{figure}[h]
\small
\newcommand{\rspace}{0cm}
\begin{align}
\shortintertext{\centering \textbf{Standard computation rules}}
\tmapp{(\tmfun{x:X}{M})}{V} &\reduces M[V/x]\label{r1}\tag{r1}\\[\rspace]
\tmlet{x}{(\tmreturn{V}{})}{M} &\reduces M[V/x]\label{r2}\tag{r2}\\[1ex]
\shortintertext{\centering \textbf{Algebraicity of signals, interrupt handlers, and awaiting}}
\tmlet{x}{(\tmopout{op}{V}{M})}{N} &\reduces \tmopout{op}{V}{\tmlet{x}{M}{N}}\label{r3}\tag{r3}\\[\rspace]
\tmlet{x}{(\tmwith{op}{y}{M_\op}{p}{N_1})}{N_2}&\reduces\notag\\[-0.5ex]
&\hspace{-1cm}\tmwith{op}{y}{M_\op}{p}{(\tmlet{x}{N_1}{N_2})} \label{r4}\tag{r4}\\[\rspace]
\tmlet{x}{(\tmawait{V}{y}{M})}{N} &\reduces \tmawait{V}{y}{(\tmlet{x}{M}{N})}\label{r5}\tag{r5}\\[1ex]
\shortintertext{\centering \textbf{Commutativity of signals with interrupt handlers}}
\tmwith{op}{x}{M_\op}{p}{\tmopout{op'}{V}{N}} &\reduces \tmopout{op'}{V}{\tmwith{op}{x}{M_\op}{p}{N}}\label{r6}\tag{r6}\\[1ex]
\shortintertext{\centering \textbf{Interrupt propagation}}
\tmopin{op}{V}{\tmreturn{W}}&\reduces \tmreturn{W}\label{r7}\tag{r7}\\[\rspace]
\tmopin{op}{V}{\tmopout{op'}{W}{M}}&\reduces \tmopout{op'}{W}{\tmopin{op}{V}{M}}\label{r8}\tag{r8}\\[\rspace]
\tmopin{op}{V}{\tmwith{op}{x}{M_\op}{p}{N}}&\reduces\tmlet{p}{M_\op[V/x]}{\tmopin{op}{V}{N}}\label{r9}\tag{r9}\\[\rspace]
\tmopin{op'}{V}{\tmwith{op}{x}{M_\op}{p}{N}} &\reduces\notag\\[-0.5ex]
&\hspace{-2.5cm}\tmwith{op}{x}{M_\op}{p}{\tmopin{op'}{V}{N}}, \text{when}~ \tm{op} \neq \tm{op'} \label{r10}\tag{r10}\\[\rspace]
\tmopin{op}{V}{\tmawait{W}{x}{M}}&\reduces\tmawait{W}{x}{\tmopin{op}{V}{M}} \label{r11}\tag{r11}\\[1ex]
\shortintertext{\centering \textbf{Unblocking awaiting with a fulfilled promise}}
\tmawait{\tmpromise{V}}{x}{M}&\reduces M[V/x]\label{r12}\tag{r12}\\[1ex]
\shortintertext{\centering \textbf{Evaluation context rule}}\notag
M\reduces N \quad \text{implies}&\quad \E[M] \reduces \E[N], \text{ where } \label{r13}\tag{r13}
\end{align}
\vspace{-4ex}
$$\E \bnfis [\;] \bnfor \tmlet{x}{\E}{N} \bnfor \tmopout{op}{V}{\E} \bnfor \tmopin{op}{V}{\E} \bnfor \tmwith{op}{x}{M_\op}{p}{\E}$$
\caption{Evaluation-contexts-based small-step operational semantics of \lambdaAEff's sequential part.}
\label{fig:seq-reduction-rules}
\end{figure}

We also note that this semantics is type-safe in the standard
sense~\cite{Wright:SynAppTypeSoundness}, as shown by Ahman and
Pretnar~\cite{Ahman:LMCS}, who proved corresponding progress and type
preservation theorems.

\subsection{Parallel Part}
\label{sec:aeff:par}

As noted in \cref{sec:intro}, to model examples such as client-server
interactions, remote code execution, and environments managing and pre-empting
threads, \lambdaAEff~extends FGCBV-style values and computations with a third
class of terms, \emph{parallel processes}, ranged over by $P, Q, \ldots$.
Processes allow computations to be run asynchronously in parallel, with the
communication between individual computations facilitated by signals and
interrupts. They are given by
\[
  P, Q
  \bnfis \tmrun M
  \,\bnfor\! \tmpar P Q
  \,\bnfor\! \tmopout{op}{V}{P}
  \,\bnfor\! \tmopin{op}{V}{P}
\]
Here, $\tmrun M$ denotes a computation being executed as a process, $\tmpar P Q$
denotes the parallel composition of two processes, $\tmopout{op}{V}{P}$ denotes
a process that has issued a signal $\op$ and proceeds as $P$, and
$\tmopin{op}{V}{P}$ denotes a process $P$ that has received an interrupt $\op$.

\emph{Well-typed processes} are classified by the judgement $\Gamma \types P :
\tyC$, where $\tyC$ is a \emph{process type} defined by the grammar  
$
  \text{$\tyC$, $\tyD$}
  \bnfis \tyrun X \o \i
  \bnfor \typar \tyC \tyD
$. Processes are then typed as
\begin{mathpar}
  \coopinfer{}{
    \Gamma \types M : \tycomp{X}{(\o,\i)}
  }{
    \Gamma \types \tmrun{M} : \tyrun{X}{\o}{\i}
  }
  \and
  \coopinfer{}{
    \Gamma \types P : \tyC \\
    \Gamma \types Q : \tyD
  }{
    \Gamma \types \tmpar{P}{Q} : \typar{\tyC}{\tyD}
  }
\end{mathpar}
with signals and interrupts typed analogously to their computation
counterparts.

The processes also come equipped with a small-step operational semantics, given
by a \emph{reduction relation} $P \reduces Q$ defined in
\cref{fig:par-reduction-rules}. The rule \eqref{r14} expresses that if a
computation can reduce on its own, then it can also do so as a process. The rule
\eqref{r15} allows signals to be propagated out of computations, after which
they keep propagating outwards, at the same time generating corresponding new
interrupts for any processes parallel to them---this is specified by the
broadcast rules \eqref{r16}--\eqref{r17}. In turn, these generated interrupts
start structurally propagating inwards into the processes they encompass, as
specified in the rules \eqref{r19}--\eqref{r20}. Once an interrupt reaches one
of the computations in a $\tmkw{run}$-leaf, the rule \eqref{r18} propagates the
interrupt into the computation, where it will continue propagating inwards using
the interrupt propagation rules for computations from
\cref{fig:seq-reduction-rules}. Finally, similarly to computations, there is an
evaluation context rule \eqref{r21} that allows execution in subprocesses. 

\begin{figure}[h]
\small
\newcommand{\rspace}{0cm}
\begin{align}
\shortintertext{\centering \textbf{Individual computations}}
M \reduces N \quad \text{implies}&\quad \tmrun{M} \reduces \tmrun{N} \label{r14}\tag{r14}\\[1ex]
\shortintertext{\centering \textbf{Signal hoisting}}
\tmrun {(\tmopout{op}{V}{M})}  &\reduces \tmopout{op}{V}{\tmrun M}\label{r15}\tag{r15}\\[1ex]
\shortintertext{\centering \textbf{Broadcasting}}
\tmpar{\tmopout{op}{V}{P}}{Q} &\reduces \tmopout{op}{V}{\tmpar{P}{\tmopin{op}{V}{Q}}}\label{r16}\tag{r16}\\[\rspace]
\tmpar{P}{\tmopout{op}{V}{Q}} &\reduces \tmopout{op}{V}{\tmpar{\tmopin{op}{V}{P}}{Q}}\label{r17}\tag{r17}\\[1ex]
\shortintertext{\centering \textbf{Interrupt propagation}}
\tmopin{op}{V}{\tmrun M} &\reduces \tmrun {(\tmopin{op}{V}{M})}\label{r18}\tag{r18}\\[\rspace]
\tmopin{op}{V}{\tmpar P Q} &\reduces \tmpar {\tmopin{op}{V}{P}} {\tmopin{op}{V}{Q}}\label{r19}\tag{r19}\\[\rspace]
\tmopin{op}{V}{\tmopout{op'}{W}{P}} &\reduces \tmopout{op'}{W}{\tmopin{op}{V}{P}}\label{r20}\tag{r20}\\[1ex]
\shortintertext{\centering \textbf{Evaluation context rule}}\notag
\hspace{3cm}
P \reduces Q \quad \text{implies} \quad \F&[P] \reduces \F[Q], \text{ where } \label{r21}\tag{r21}
\end{align}
\vspace{-4ex}
$$\F
  \bnfis [\;]
  \bnfor \tmpar \F Q \bnfor\! \tmpar P \F
  \bnfor \tmopout{op}{V}{\F}
  \bnfor \tmopin{op}{V}{\F}$$
\caption{Evaluation-contexts-based small-step operational semantics of \lambdaAEff's parallel part.}
\label{fig:par-reduction-rules}
\end{figure}

The parallel part of \lambdaAEff~also enjoys type-safety, again proved via
progress and type preservation~\cite{Ahman:LMCS}. The interesting caveat of the
type-safety for processes is that the preservation result has to account
for the generation of new interrupts in types by the broadcast rules. 

%% file: normalisation-sequential.tex
\section{Strong Normalisation for Sequential Computations}
\label{sec:sn-seq}

In this section, we prove that the sequential part of the \lambdaAEff-calculus
that we reviewed in \cref{sec:aeff:seq} is strongly normalising, i.e., we 
show that there are no infinite reduction sequences starting from well-typed
computations. As noted in \cref{sec:intro}, our normalisation proof is based on
the compositional Girard-Tait-style type-directed reducibility approach to
proving strong normalisation of higher-order languages~\cite{proofs,tait}, and
its extension to effectful programming languages via the $\top\top$-lifting
technique proposed by Lindley and Stark~\cite{Lindley:TopTopLifting}.

\subsection{Reducibility and $\top\top$-Lifting-Based Strong Normalisation Proofs}

We recall how the Girard-Tait-style type-directed reducibility approach to
proving strong normalisation is structured at a high level. The approach is
originally due to Tait~\cite{tait}, and later improved and extended by
Girard~\cite{proofs}. We follow the style of the latter presentation.

First, we remind the reader that a term $M$ is \emph{strongly normalising} when
there does not exist an infinite reduction sequence $M \reduces M' \reduces
\ldots$ starting from $M$. A calculus is strongly normalising when all its
well-typed terms are. We denote the set of all strongly normalising terms of
type $X$ with $\SN_{\!X}$, often omitting the $X$.
In proofs, it is convenient to give $\SN_{\!X}$ an
equivalent \emph{inductive characterisation} as the least set 
closed under the following rule: 
\[
  \inferrule*[]{\forall N .\, M \reduces N \implies N \in \SN_{\!X}}{M \in \SN_{\!X}}
\]

Proving strong normalisation of higher-order languages directly by induction on
the syntax (or typing derivations) does not work well. Not only can the number
of redexes grow after some steps, 
but also knowing that a term of some type is strongly normalising usually does not
immediately imply that its subterms are also strongly normalising (e.g., knowing
that a function abstraction is strongly normalising does not imply that
its body is as well).

The Girard-Tait-style approach addresses this situation by defining
a notion of \emph{reducibility} at every type, a form of unary logical
relation~\cite{plotkin:logical-relations}. Strong normalisation is then proved
in two steps: one proves that (i) every reducible term is strongly normalising,
and that (ii) every well-typed term is reducible at its type. The
earlier problems are solved by the reducibility relation being
amenable to induction on the structure of types, allowing one to abstractly
specify that subterms of reducible (strongly normalising) terms have the same
property.

The $\top\top$-lifting technique proposed by Lindley and
Stark~\cite{Lindley:TopTopLifting} extends this approach to effectful languages.
As the types of computations typically do not characterise the structure of
computations, only that of return values, one cannot define reducibility for
computations simply over the structure of types as one does for values. Lindley
and Stark's proposal is to define the notion of reducibility for computations in
two steps: (i) a computation is reducible when it is strongly normalising in any
reducible term context (called continuations in~\cite{Lindley:TopTopLifting}),
and (ii) such a continuation is reducible if it yields a strongly normalising
computation when filled with any reducible value at the corresponding type. The
continuations model the environment the computation is executing in, and this
definition of reducibility then naturally extends reducibility of (return)
values to computations and their environments.

In this section, we show how these ideas apply to the \lambdaAEff-calculus. We
note that the work of Lindley and Stark is presented in the context of Moggi's
computational metalanguage~\cite{moggi:notions}, while the \lambdaAEff-calculus
is based on FGCBV~\cite{Levy:FGCBV}. The difference in these calculi is, however,
superficial, and the $\top\top$-lifting-based approach straightforwardly
extends to the latter~\cite[§2]{sobolev:bsc}.

For proving strong normalisation for the sequential part of \lambdaAEff, we
consider a \emph{skeletal version} of it, in which we omit effect annotations.
This is to emphasise that this normalisation result does not depend on
effect-typing. The \emph{skeletal types} are given by the grammar
\[
  X,Y \bnfis \tybase \bnfor \tyfun{X}{Y} \bnfor \typromise{X}
\]
Typing rules are similar to \cref{sec:aeff:seq}.
While this normalisation proof does not rely on effect annotations, they
still have an important role in the \lambdaAEff-calculus as static
specifications (for users) about which signals a program can issue and which
interrupts it might handle.
Furthermore, they are key to proving strong normalisation for \lambdaAEff's parallel part in \cref{sec:sn-parallel}.

\subsection{Continuations}
\label{sec:sn-continuations}

Following Lindley and Stark~\cite{Lindley:TopTopLifting}, we first define a
notion of \emph{continuations} $K$, i.e., term contexts that model the
environments in which computations execute, given by the grammar:
\[
  K \bnfis \Id \bnfor K \circ (x)N \bnfor K \circ \copin{op}{V}
\]
Their intuitive meaning is best understood by how they are \emph{applied to
computations}:
\[
  \begin{array}{r c l}
    \Id \capp M & \defeq & M,
    \\
    (K \circ (x)N) \capp M & \defeq & K \capp (\tmlet{x}{M}{N}),
    \\
    (K \circ \copin{op}{V}) \capp M & \defeq & K \capp \tmopin{op}{V}{M}.
  \end{array}
\]
The first two cases are already present in Lindley and Stark's work, and they
express that after executing $M$ there is either nothing more to be done, or
that $N$ is sequentially executed after $M$ returns. The third case is new for
\lambdaAEff: it expresses that the environment has an interrupt ready to
propagate into $M$. This allows us to prove that computations are strongly
normalising for any finite number of interrupts propagated to them by the
environment.

Observe that the continuations $K$ are just an inside-out reformulation of a
subset of \lambdaAEff's evaluation contexts $\E$, following a
tradition in $\top\top$-closed-relations-based techniques to aid
reasoning~\cite{pitts:parapoly}.
They do not include the remaining two cases of \lambdaAEff's evaluation
contexts: signals and interrupt handlers. This is because these evaluation
contexts express actions taken by computations, and not by the environment.
Semantically speaking, the difference is that signals and interrupt handlers
behave like algebraic operations (i.e., constructors of computations), while
sequencing and interrupts behave like effect handling (i.e., eliminators for
computations, modelling the environments in which computations execute).

While the evaluation contexts $\E$ naturally do not include $\tmkw{await}$
because it has a blocking semantics, we find it useful for defining the notion
of reducibility in \cref{sec:sn-seq:reducibility} to introduce an auxiliary
notion of \emph{await continuations} $K^{\langle\rangle}$, which are defined by
a single case as ${K^{\langle\rangle} \bnfis K \circ \tmpromise{x}N}$ and which
are \emph{applied to values}, instead of computations, as ${(K \circ
\tmpromise{x}N) \capp V} \defeq K \capp (\tmawait{V}{x}{N})$. 
We use them to simplify reasoning about how computations $N$ awaiting a
promise to be fulfilled behave in environments $K$, similarly to how Lindley and
Stark structure reasoning about the sum type~\cite{Lindley:TopTopLifting} (see
also \cref{sec:sn-seq:sums-products}).

We characterise \emph{well-typed continuations} by the judgement $\Gamma \types
K : X \multimap Y$, defined by
\begin{mathpar}
  \coopinfer{}{
    \phantom{...}
  }{
    \Gamma \types \Id : X \multimap X
  }
  \and
  \coopinfer{}{
    \Gamma, x \of X \types N : Y \\
    \Gamma \types K : Y \multimap Z
  }{
    \Gamma \types K \circ (x)N : X \multimap Z
  }
  \and
  \coopinfer{}{
    \op \in \Sigma \\
    \Gamma \types V : A_\op \\
    \Gamma \types K : X \multimap Y 
  }{
    \Gamma \types K \circ \copin{op}{V} : X \multimap Y
  }
\end{mathpar}
Well-typed await continuations are characterised by an analogous rule, which we
omit here.

It is straightforward to see that the application $K \capp M$ of a continuation
${\Gamma \types K : X \multimap Y}$ to a computation $\Gamma \types M : X$
preserves typing.
We also define the \emph{length $| K |$ of a continuation}, as the number of
compositions $\circ$ in $K$. This is a measure of the size of the environment $K$. 

\subsection{Reducibility}
\label{sec:sn-seq:reducibility}

We now turn to defining the notion of \emph{reducibility}, as type- and
context-indexed predicates $\VRed_X^\Gamma$, $\CRed_X^\Gamma$, and
$\KRed_X^\Gamma$ over well-typed values $\Gamma \types V : X$, computations
$\Gamma \types M : X$, and continuations $\Gamma \types K : X \multimap Y$,
respectively. 
When a well-typed value, computation, or continuation satisfies the respective reducibility
predicate, we say that it is \emph{reducible}.

Compared to Lindley and Stark, we do not define these predicates
pointwise for contexts $\Gamma$, but instead in the style of Kripke's possible
world semantics~\cite{kripke:semantics}. 
This is needed to reason about interrupt handlers $\tmwith{op}{x}{M}{p}{N}$ in
\cref{prop:sn-seq:interrupt-handlers-sn,prop:sn-seq:interrupt-handlers-reducible},
because we need to weaken a reducible continuation $K \in \KRed_Y^\Gamma$ with $p$ to apply it to
$N$.\footnote{The Bachelor's thesis~\cite[§3]{sobolev:bsc} of one of the authors
demonstrates how a non-Kripke-style definition of reducibility can be used
instead, with the subtle caveat that then one has to extend \lambdaAEff-values
with new constants $\star : \typromise{X}$ that intuitively model promises that
have not yet been fulfilled at a given time.} To give such a Kripke-style
definition of reducibility, we first define the \emph{set of renamings}
$\Ren{\Gamma}{\Gamma'}$ as
\[
  \Ren{\Gamma}{\Gamma'} \defeq \{ \rho : Vars(\Gamma) \to Vars(\Gamma') \mid \forall\, x \of X \in \Gamma .\, \rho(x) \of X \in \Gamma' \}.
\]
We write $\Gamma' \types V[\rho] : X$ for the application of the \emph{action of
a renaming} $\rho \in \Ren{\Gamma}{\Gamma'}$ on a value $\Gamma \types V : X$, and
similarly for computations and continuations. For instance, we write $\mathit{wk} \in
\Ren{\Gamma}{(\Gamma, x \of X)}$ and $\mathit{lift}(\rho) \in \Ren{(\Gamma, x \of
X)}{(\Gamma', x \of X)}$ for the standard \emph{weakening} and
\emph{lifting renamings}, where $\rho \in \Ren{\Gamma}{\Gamma'}$.
We sometimes leave the action of renamings implicit.

The \emph{reducibility predicates} are then defined mutually inductively as
follows:
\[
  \begin{array}{@{}l c l}
    V \in \VRed_\tybase^\Gamma & \iff & \top,
    \\[0.5ex]
    V \in \VRed_{\tyfun{X}{Y}}^\Gamma & \iff & \forall\, \Gamma' .\, \forall\, \rho \in \Ren{\Gamma}{\Gamma'} .\, \forall\, W \in \VRed_X^{\Gamma'} .\, (V[\rho])\, W \in \CRed_Y^{\Gamma'},
    \\[0.5ex]
    V \in \VRed_{\typromise{X}}^{\Gamma} & \iff & \forall\, \Gamma' .\, \forall\, \rho \in \Ren{\Gamma}{\Gamma'} .\, \forall\, K^{\langle\rangle} \in \ARed_X^{\Gamma'} .\, K^{\langle\rangle} \capp (V[\rho]) \in \SN,
    \\[1ex]
    M \in \CRed_X^\Gamma & \iff & \forall\, \Gamma' .\, \forall\, \rho \in \Ren{\Gamma}{\Gamma'} .\, \forall\, K \in \KRed_X^{\Gamma'} .\, K \capp (M[\rho]) \in \SN,
    \\[1ex]
    K \in \KRed_X^\Gamma & \iff & \forall\, \Gamma' .\, \forall\, \rho \in \Ren{\Gamma}{\Gamma'} .\, \forall\, V \in \VRed_X^{\Gamma'} .\, (K[\rho]) \capp (\tmreturn{V}) \in \SN.
  \end{array}
\]
As noted in \cref{sec:sn-continuations}, for the promise type $\typromise{X}$,
the predicate $\VRed_{\typromise{X}}^{\Gamma}$ uses an auxiliary notion of await
continuations $K^{\langle\rangle}$. The corresponding reducibility predicate is
defined as 
\[
  \begin{array}{@{}l c l}
    K^{\langle\rangle} \in \ARed_X^\Gamma & \iff & \forall\, \Gamma' .\, \forall\, \rho \in \Ren{\Gamma}{\Gamma'} .\, \forall\, V \in \VRed_X^{\Gamma'} .\, (K^{\langle\rangle}[\rho]) \capp \tmpromise{V} \in \SN.
  \end{array}
\]
In other words, await continuations model environment-computation pairs which
are awaiting a promise to be fulfilled, and the reducible ones are those that
are strongly normalising when supplied with a fulfilled promise containing a
reducible value of the appropriate type.

Interestingly, despite the additional features of \lambdaAEff, the overall structure of
these predicates is identical to Lindley and Stark's work, apart from the
Kripke-style indexing.
In particular, while in \lambdaAEff~the interrupt handlers
$\tmwith{op}{x}{M_\op}{p}{N}$ behave similarly to scoped algebraic
effects~\cite{Pirog:ScopedOperations}, in that they carry a computation $M_\op$
which is executed only when a corresponding interrupt arrives, we do not have to
mention these $M_\op$s in reducibility, and thus we do not have to use
techniques such as step-indexing~\cite{appel-mcallester:step-indexing} to ensure
that the reducibility predicates are well-founded. The Kripke-style indexing is
not because of handler code $M_\op$, but instead because of additional
variables being bound in the non-blocking continuation $N$.

Note that this use of Kripke-style indexing due to the variable binding and
non-blocking semantics of interrupt handlers is somewhat reminiscent of how
call-by-need $\lambda$-calculi are modelled
logically~\cite{herbelin:calculus-of-expandable-stores,miquey:realizability-of-call-by-need},
where computations also have to execute in open contexts (that can only grow
over time) because of lazy let-bindings. However, one difference is that when
the execution of a call-by-need $\lambda$-calculus reaches a variable from such
a context, it forces the corresponding let-binding to execute. Meanwhile, when
the execution of \lambdaAEff~ends up awaiting a promise-typed variable bound
by an interrupt handler, the whole computation blocks until some external
interrupt triggers that interrupt handler and the promise gets fulfilled.

The Kripke-style definition of these reducibility predicates gives us the 
next result.

\begin{proposition}
  \label{prop:sn-seq:reducibility-renaming}
  If $V \in \VRed_X^\Gamma$ and $\rho \in \Ren{\Gamma}{\Gamma'}$, then we
  have $V[\rho] \in \VRed_X^{\Gamma'}$, and analogously for the other
  reducibility predicates for computations and continuations.
\end{proposition}

Furthermore, we can prove that application of continuations preserves
reducibility.

\begin{proposition}
  If ${\Gamma \types\! K : X \!\multimap\! Y}$, ${K \!\in\! \KRed_X^\Gamma}$, and
  ${M \!\in\! \CRed_X^\Gamma}$, then ${K \!\capp\! M \!\in\! \CRed_Y^\Gamma}$.
\end{proposition}

\subsection{Strong Normalisation}

As is standard in Girard-Tait-style normalisation proofs, as a first step we
prove that strong normalisation follows (straightforwardly) from a computation
being reducible.

\begin{theorem}
  \label{thm:sn-seq:reducibility-sn}
  If $\Gamma \types M : X$, such that $M \in \CRed_X^\Gamma$, then $M \in \SN_{\!X}$.
\end{theorem}

\begin{proof}
  The proof is identical to that of Lindley and Stark. By the definition of
  $\CRed_X^\Gamma$, we have $\Id \capp M = M \in \SN_{\!X}$ if we can show
  $\Id \in \KRed_X^\Gamma$. But by the definition of $\KRed_X^\Gamma$, it
  suffices to show that $\Id \capp (\tmreturn{V}) \in \SN_{\!X}$, for all $V \in
  \VRed_X^\Gamma$, which is trivially true.
\end{proof}

The rest of the normalisation proof amounts to showing that all well-typed
computations are reducible. It is customary to split the proof of this result
(\cref{thm:sn-seq:sn}) into propositions based on the typing rules of
computations. We only discuss the cases involving \lambdaAEff-specific
computations---details of other cases can be found in our Agda formalisation, or in
\cite{Lindley:TopTopLifting,sobolev:bsc}.

As is standard, we first prove that \emph{variables} are reducible, by induction
on their types.

\begin{proposition}
  \label{prop:sn-seq:variables-reducible}
  If $\Gamma = x_1 \of X_1, \ldots, x_n \of X_n$, then $x_i \in
  \VRed_{X_i}^\Gamma$, for all $x_i \of X_i \in \Gamma$.
\end{proposition}

Next, we prove below in \cref{prop:sn-seq:signals-reducible} that signals
$\tmopout{op}{V}{M}$ are reducible. For this, however, we first need to establish
an auxiliary result about signals and strong normalisation, from which
\cref{prop:sn-seq:signals-reducible} follows by straightforward application of
the definition of $\CRed_X^\Gamma$.

For $M \in \SN_{\!X}$, we write $\mathit{max}(M)$ for the length of the longest reduction
sequence starting from $M$, which exists since a \lambdaAEff-computation has
only finitely many one-step reducts.

\begin{proposition}
  \label{prop:sn-seq:signals-sn}
  Consider a continuation $\Gamma \types K : X \multimap Y$, an operation $\op
  \in \Sigma$, a value $\Gamma \types V : A_\op$, and a computation $\Gamma
  \types M : X$. If $K \capp M \in \SN_{\!Y}$, then $K \capp
  (\tmopout{op}{V}{M}) \in \SN_{\!Y}$.
\end{proposition}

\begin{proof}
  The proof proceeds by induction on the sum $| K | + \mathit{max}(K \capp M)$, and we
  show that all reducts of the computation $K \capp (\tmopout{op}{V}{M})$ are strongly
  normalising.
\end{proof}

\begin{proposition}
  \label{prop:sn-seq:signals-reducible}
  If $\Gamma \types V : A_\op$ and $M \in \CRed_X^\Gamma$, then
  $\tmopout{op}{V}{M} \in \CRed_X^\Gamma$.
\end{proposition}

Next, we prove that interrupts $\tmopin{op}{V}{M}$ are reducible, 
again using auxiliary results. 

\begin{proposition}
  \label{prop:sn-seq:interrupts-sn}
  Consider a continuation $\Gamma \types K : X \multimap Y$, an operation $\op
  \in \Sigma$, and values $\Gamma \types V : A_\op$ and $\Gamma \types W :
  X$. If $K \capp (\tmreturn{W}) \in \SN_{\!Y}$, then $K \capp
  (\tmopin{op}{V}{\tmreturn{W}}) \in \SN_{\!Y}$.  
\end{proposition}

\begin{proof}
  The only possible reduction step is \eqref{r7} under $K$, resulting in $K
  \capp (\tmreturn{W})$.
\end{proof}

\begin{proposition}
  \label{prop:sn-seq:interrupts-reducibility}
  Consider a continuation $\Gamma \types K : X \multimap Y$, an operation $\op
  \in \Sigma$, and a value $\Gamma \types V : A_\op$. If $K \in \KRed_X^\Gamma$,
  then we also have that $K \circ \copin{op}{V} \in \KRed_X^\Gamma$.
\end{proposition}

\begin{proof}
  By direct application of \cref{prop:sn-seq:interrupts-sn} and the definition of
  $\KRed_X^\Gamma$.
\end{proof}

\begin{proposition}
  \label{prop:sn-seq:interrupts-reducible}
  If $\Gamma \types V : A_\op$ and $M \in \CRed_X^\Gamma$, then
  $\tmopin{op}{V}{M} \in \CRed_X^\Gamma$.
\end{proposition}

Next, we prove that $\tmawait{V}{x}{N}$ is reducible. Observe that as the
(blocked) continuation $N$ can contain an additional variable $x$, the
assumptions of the propositions below encode a condition as if we were
considering $N$ as a function $\tmfun{x \of X}{N}$.

\begin{proposition}
  \label{prop:sn-seq:awaiting-sn}
  Consider a continuation $\Gamma \types K : Y \multimap Z$, a value $\Gamma
  \types V : X$, and $\Gamma, x \of X \types N : Y$. If $K \capp N[V/x] \in
  \SN_{\!Z}$, then $K \capp (\tmawait{\tmpromise{V}}{x}{N}) \in \SN_{\!Z}$.
\end{proposition}

\begin{proof}
  The proof proceeds by induction on $| K |$, and we show that all reducts of the computation $K \capp (\tmawait{\tmpromise{V}}{x}{N})$ are strongly normalising.
\end{proof}

\begin{proposition}
  \label{prop:sn-seq:awaiting-reducibility}
  Consider a continuation $\Gamma \types K : Y \multimap Z$ and a computation
  term $\Gamma, x \of X \types N : Y$. If $K \in \KRed_Y^\Gamma$ and
  $N[\mathit{lift}(\rho)][W/x] \in \CRed_Y^{\Gamma'}$, for all $\Gamma'$, $\rho \in
  \Ren{\Gamma}{\Gamma'}$, and $W \in \VRed_X^{\Gamma'}$, then $K \circ
  \tmpromise{x}N \in \ARed_X^\Gamma$.
\end{proposition}

\begin{proposition}
  \label{prop:sn-seq:awaiting-reducible}
  Consider a value $\Gamma \types V : \typromise{X}$ and a computation $\Gamma,
  x \of X \types N : Y$. If we have $V \in \VRed_{\typromise{X}}^\Gamma$ and
  $N[\mathit{lift}(\rho)][W/x] \in \CRed_Y^{\Gamma'}$, for all $\Gamma'$, $\rho \in
  \Ren{\Gamma}{\Gamma'}$, and $W \in \VRed_X^{\Gamma'}$, then we also have that
  $\tmawait{V}{x}{N} \in \CRed_Y^\Gamma$.
\end{proposition}

Finally, we prove that $\tmwith{op}{x}{M_\op}{p}{N}$ is reducible. This is the
most involved case for the \lambdaAEff-specific computations because the
continuation $N$ is non-blocking and can contain the promise-typed variable $p$.
In particular, observe that to be able to reason about $K[\mathit{wk}] \capp N$ in
\cref{prop:sn-seq:interrupt-handlers-sn} below, it is crucial that we have
defined the reducibility predicates in Kripke style, because $\Gamma \types K
: Y \multimap Z$, but $\Gamma, p \of \typromise{X} \types K[\mathit{wk}] \capp N : Z$.

First, in \cref{prop:sn-seq:interrupt-handlers-auxiliary-results} we prove some
auxiliary results about the reducibility of continuations. Then, in
\cref{prop:sn-seq:interrupt-handlers-sn} we do most of the heavy lifting,
showing that if the handler code and the non-blocking continuation of the
interrupt handler are strongly normalising, then so is the interrupt handler. We
then use these results to prove reducibility in \cref{prop:sn-seq:interrupt-handlers-reducible}.

\begin{proposition}
  \label{prop:sn-seq:interrupt-handlers-auxiliary-results}
  For suitably typed continuations, values, and computations, we have:
  \begin{itemize}
    \item If $M \reduces N$ and $K \circ (x)M \in \KRed_X^\Gamma$, then $K \circ (x)N
    \in \KRed_X^\Gamma$.
    \item If $K \circ (y)N \circ (x)M \in \KRed_X^\Gamma$, then $K \circ
    (x)(\tmlet{y}{M}{N[\mathit{lift}(\mathit{wk})]}) \in \KRed_X^\Gamma$.
    \item If $K \circ (x)(\tmopout{op}{V}{M}) \in \KRed_X^\Gamma$, then $K \circ
    (x)M \in \KRed_X^\Gamma$.
    \item If $K \circ \copin{op}{V} \circ (x)N \in \KRed_X^\Gamma$, then $K
    \circ (x)(\tmopin{op}{V[\mathit{wk}]}{N}) \in \KRed_X^\Gamma$.
  \end{itemize}
\end{proposition}

\begin{proposition}
  \label{prop:sn-seq:interrupt-handlers-sn}
  Consider a continuation $\Gamma \types K : Y \multimap Z$, an operation $\op
  \in \Sigma$, and computations $\Gamma, x \of A_\op \types M_\op :
  \typromise{X}$ and $\Gamma, p \of \typromise{X} \types N : Y$. If $K \circ
  (p)N \in \KRed_{\typromise{X}}^\Gamma$, $K[\mathit{wk}] \capp N \in \SN_{\!Z}$, and
  $M_\op[\mathit{lift}(\rho)][V/x] \in \CRed_{\typromise{X}}^{\Gamma'}$, for all
  $\Gamma'$, $\rho \in \Ren{\Gamma}{\Gamma'}$, and $V \in
  \VRed_{A_\op}^{\Gamma'}$, then $K \capp \tmwith{op}{x}{M_\op}{p}{N} \in
  \SN_{\!Z}$.
\end{proposition}

\begin{proof}
  The proof is by induction on the lexicographic order of $| K | + \mathit{max}(K[\mathit{wk}]
  \capp N)$ and the structure of $N$, and we show that all reducts of $K \capp
  \tmwith{op}{x}{M_\op}{p}{N}$ are strongly normalising. For details, we refer
  the reader to \cref{appendix:prop:sn-seq:interrupt-handlers-sn}, or to our Agda
  code.  
\end{proof}

\begin{proposition}
  \label{prop:sn-seq:interrupt-handlers-reducible}
  Consider an operation $\op \in \Sigma$, and computations $\Gamma, x \of A_\op
  \types M_\op : \typromise{X}$ and $\Gamma, p \of \typromise{X} \types N : Y$.
  If we have $M_\op[\mathit{lift}(\rho)][V/x] \in \CRed_{\typromise{X}}^{\Gamma'}$, for
  all $\Gamma'$, $\rho \in \Ren{\Gamma}{\Gamma'}$, and $V \in
  \VRed_{A_\op}^{\Gamma'}$, and if we also have $N[\mathit{lift}(\rho)][W/p] \in
  \CRed_Y^{\Gamma'}$, for all $\Gamma'$, $\rho \in \Ren{\Gamma}{\Gamma'}$, and
  $W \in \VRed_{\typromise{X}}^{\Gamma'}$, then we have that $\tmwith{op}{x}{M_\op}{p}{N} \in
  \CRed_Y^\Gamma$.
\end{proposition}

We can now prove the overall reducibility result for well-typed computations.
The proof is a simple induction on the typing derivation, using the propositions
we proved above.

\begin{theorem}
  \label{thm:sn-seq:fundamental-theorem}
  Consider a well-typed computation $\Gamma \types M : Y$, where $\Gamma = x_1
  \of X_1, \ldots, x_n \of X_n$, and well-typed values $\Gamma' \types V_i :
  X_i$, for all $x_i \of X_i \in \Gamma$. If $V_i \in \VRed_{X_i}^{\Gamma'}$,
  for all $x_i \of X_i \in \Gamma$, then $M[V_1/x_1, \ldots, V_n/x_n] \in
  \CRed_Y^{\Gamma'}$. An analogous result also holds for well-typed values.
\end{theorem}

\begin{corollary}
  \label{cor:sn-seq:reducibility}
  As $x_i \in \VRed_{X_i}^{\Gamma}$ by \cref{prop:sn-seq:variables-reducible},
  we have that $M \in \CRed_Y^\Gamma$.
\end{corollary}

Finally, we can compose \cref{cor:sn-seq:reducibility} with
\cref{thm:sn-seq:reducibility-sn} to prove strong normalisation.

\begin{theorem}
  \label{thm:sn-seq:sn}
  Every well-typed computation term $\Gamma \types M : Y$ is strongly normalising. 
\end{theorem}

We can relate this result to the \lambdaAEff-calculus from \cref{sec:aeff} by
defining an erasure $| - |$ from the types and contexts of \cref{sec:aeff} to
the skeletal types and contexts of this section.

\begin{corollary}
  \label{cor:sn-seq:simulation}
  If $\Gamma \types M : \tycomp{X}{(\o,\i)}$, then $| \Gamma | \types M : | X |$, and 
  thus $M$ is strongly normalising.
\end{corollary}

\subsection{Unit, Product, and Sum Types}
\label{sec:sn-seq:sums-products}

We note that unit, product, and sum types can be added to \lambdaAEff~and the
normalisation proof in a standard way~\cite{Lindley:TopTopLifting}. We make
use of the unit and sum types in the next section.

For the \emph{unit type} $\tyunit$, we add a value $\tmunit$, and define its
reducibility like for base types, to always hold.
For the \emph{sum type} $\tysum{X}{Y}$, we extend values with injections
$\tminl[Y]{V}$ and $\tminr[X]{V}$, and computations with pattern-matching
$\tmmatch[]{V}{\tminl[]{x} \mapsto M , \tminr[]{y} \mapsto N}$, with the evident
reduction rules. The \emph{reducibility for $X + Y$} is then defined analogously
to how we defined it for the promise type in \cref{sec:sn-seq:reducibility}. We
say that $V \in \VRed_{\tysum{X}{Y}}^{\Gamma}$ iff $\forall\, \Gamma' .\,
\forall\, \rho \in \Ren{\Gamma}{\Gamma'} .\, \forall\, K^{+} \in
\SRed_{X,Y}^{\Gamma'} .\, K^{+} \capp (V[\rho]) \in \SN$, where $K^{+} \!\bnfis\! K
\circ ((x)M,(y)N)$, and $(K \circ ((x)M,(y)N)) \capp V \defeq K \capp
(\tmmatch[]{V}{\tminl[]{x} \mapsto M , \tminr[]{y} \mapsto N})$. The
reducibility for such \emph{sum continuations} is defined as $K^{+}\in
\SRed_{X,Y}^{\Gamma}$ iff $\forall\, \Gamma' .\, \forall\, \rho \in
\Ren{\Gamma}{\Gamma'} .\, \forall\, V\in \VRed_{X}^{\Gamma'} .\, (K^{+}[\rho]) \capp
(\tminl[Y]{V}) \in \SN$ and $\forall\, \Gamma' .\, \forall\, \rho \in
\Ren{\Gamma}{\Gamma'} .\, \forall\, W\in \VRed_{Y}^{\Gamma'} .\, (K^{+}[\rho]) \capp
(\tminr[X]{W}) \in \SN$. 

%% file: normalisation-reinstallable.tex
\section{Strong Normalisation for Reinstallable Interrupt Handlers}
\label{sec:sn-reinstallable}

Although the \lambdaAEff-computations from \cref{sec:aeff:seq} are strongly
normalising, we cannot write very interesting asynchronous examples using them.
For instance, in a client-server scenario, there is no way for the server to
react to an unbounded number of requests from clients, nor can a thread be
stopped and started an unbounded number of times by the environment. Any such
iterative behaviour has a statically fixed bound given by installed interrupt
handlers.

To address this shortcoming, Ahman and Pretnar~\cite{Ahman:LMCS} extended
\lambdaAEff~with a capability to reinstall interrupt handlers. They conjectured
that the sequential part of the resulting calculus is strongly normalising. In
this paper, we first demonstrate that their conjecture does not hold
(\cref{sec:sn-reinstallable:non-sn}), and then propose a variant that is
strongly normalising (\cref{sec:sn-reinstallable:sn}).

\subsection{Failure of Strong Normalisation for Reinstallable Interrupt Handlers}
\label{sec:sn-reinstallable:non-sn}

To reintroduce a controlled amount of recursive behaviour into
\lambdaAEff-computations, Ahman and Pretnar~\cite{Ahman:LMCS} extended interrupt
handlers to a form $\tmwithre{op}{x}{r}{M_\op}{p}{N}$, where the handler code
$M_\op$ can contain an additional (function-typed) variable $r$, whose
application results in the same interrupt handler being
\emph{reinstalled}, as expressed by the reduction rule
\[
  \begin{array}{l@{~~} l}
    & \tmopin{op}{V}{\tmwithre{op}{x}{r}{M_\op}{p}{N}} 
    \\[1ex]
    \reduces & \tmlet{p}{M_\op\big[V/x,{(\tmfun{}{\tmwithre{op}{x}{r}{M_\op}{q}{\tmreturn q}})/r}\big]}{\tmopin{op}{V}{N}}.
  \end{array}
\]
For example, with such interrupt handlers one can write recursive server processes as
\[
  \tmwithre{request}{x}{r}{(\tmlet{y}{M_{\text{process}}}{\tmopout{response}{y}{\tmapp{r}{()}}})}{p}{\ldots}, 
\]
where an interrupt handler is used to wait for $\tm{request}$s from clients, which 
are processed by $M_{\text{process}}$; the result is sent back to the 
client in a $\mathsf{response}$ signal, and the application $\tmapp{r}{()}$ is 
used to reinstall the same interrupt handler, to wait for further 
$\tm{request}$s from clients.

As the recursive reinstalling of interrupt handlers is only triggered by
incoming interrupts, and there are always only finitely many of them surrounding
any given interrupt handler in a sequential computation, Ahman and Pretnar
conjectured that the sequential part of this extension of \lambdaAEff~ought to be
strongly normalising. Unfortunately, this is not the case. 

For one concrete counterexample, consider the computation 
\[
M_1 \defeq \tmopin{op}{V}{\tmwithre{op}{x}{r}{\tmopin{op}{x}{\tmapp{r}{()}}}{p}{\tmreturn{p}}}, 
\]
where in the handler code an interrupt is wrapped around the reinstalling
application $\tmapp{r}{()}$. One can show that $M_1$ is well-typed for an effect
annotation $(\emptyset,\iota)$, where $\iota$ is given (co)recursively as
$\iota = \{\tm{op}\mapsto(\emptyset,\{\tm{op}\mapsto(\emptyset,\dotsc)\})\}$, using the
coinductive definition of $I$. On the other hand, when we look at how this 
computation executes, we see that the subterm $\tmopin{op}{x}{\tmapp{r}{()}}$
leads to the same interrupt handler for $\op$ being triggered and reinstalled infinitely often.

For another counterexample, consider the computation
\[
  M_2 \defeq \tmopin{op}{V}{\tmwithre{op}{x}{r}{\tmreturn{\tmpromise{\tmfun{}{N}}}}{p}{(\tmawait{p}{y}{\tmapp{y}{()}})}}, 
\]
where $N \defeq \tmlet{p}{\tmapp{r}{()}}{(\tmawait{p}{y}{\tmapp{y}{()}})}$. This
computation is well-typed, but also exhibits infinite reduction behaviour.
This time the culprit is the leaking of the reinstall variable $r$ in the body 
of the function $\tmpromise{\tmfun{}{N}}$ into the continuation of the 
interrupt handler, where the reinstalled interrupt handler is triggered again
by the original interrupt.

\subsection{Strongly Normalising Variant of Reinstallable Interrupt Handlers}
\label{sec:sn-reinstallable:sn}

If we look at the counterexamples listed above, we see that the failure of strong
normalisation is caused by the freedom that the reinstallable interrupt handlers
proposed by Ahman and Pretnar give to the programmer for where and how the reinstall
variable $r$ can be used.
In particular, they allow $r$ to appear under interrupts in $M_\op$, and to be
leaked into $N$.
Based on these observations, we propose a restricted variant of \emph{reinstallable
interrupt handlers}:
\begin{mathpar}
  \coopinfer{}{
    ({\o'} , {\i'}) \order{O\times I} \i\, (\op)  \\
    \Gamma, x \of A_\op \types M_\op : \tycomp{\tysum{\typromise X}{\tyunit}}{(\o',\i')} \\
    \Gamma, p \of \typromise X \types N : \tycomp{Y}{(\o,\i)}
  }{
    \Gamma \types \tmwith{op}{x}{M_\op}{p}{N} : \tycomp{Y}{(\o,\i)}
  }
\end{mathpar}

Instead of modelling the reinstalling capability with a first-class variable, we
use sum types to encode it in the return type of the handler code $M_\op$.
Returning in the left injection of the sum type models the interrupt
handling finishing without reinstalling, and returning in the right injection
models the interrupt handler getting reinstalled. This behaviour is summarised by the following
reduction rule for triggering reinstallable interrupt handlers:
\[
  \begin{array}{l@{~~} l}
    & \tmopin{op}{V}{\tmwith{op}{x}{M_\op}{p}{N}} 
    \\[1ex]
    \reduces & \tmlet{p}{\big(\tmlet{y}{M_\op[V/x]}{\tmmatch{y}{\tminl{z}{}\mapsto \tmreturn{z},\tminr{w}{}\mapsto R}{}}\big)}{\tmopin{op}{V}{N}},
  \end{array}
\]
where $R = \tmwith{op}{x}{M_\op}{q}{\tmreturn{q}}$ \emph{reinstalls} the
interrupt handler when $M_\op$ returns $\tminr[]{\tmunit}$. Observe that this
way it is impossible both to wrap interrupts around reinstalled interrupt
handlers in $M_\op$, and to leak the reinstalling capability into $N$.

The upside of this definition is that we can still express all the interesting
recursive examples of Ahman and Pretnar~\cite{Ahman:LMCS}, while the sequential
part of the calculus remains strongly normalising. 
For instance, we can write the server process example in this extension as
\[
  \tmwith{request}{x}{(\tmlet{y}{M_{\text{process}}}{\tmopout{response}{y}{\tmkw{reinstall}}})}{p}{\ldots}
\]
and similarly for their other examples, in particular, including pre-emptive
multi-threading:
\[
  \begin{array}{l}
    \tmwithbig{stop}{x}{\\
    \quad \tmwith{go}{y}{\tmkw{finish}\; \tmpromise{\tmunit}}{p}{\\
    \quad \tmawait{p}{z}{\tmkw{reinstall}}}\\
    }{q}{\tmreturn{q}}
  \end{array}
\]
where we write $\tmkw{finish}\;V$ for $\tmreturn{(\tminl{V})}$ and
$\tmkw{reinstall}$ for $\tmreturn{(\tminr{\tmunit})}$ in the handler code.

At a high level, the structure of the normalisation proof remains the same as in
\cref{sec:sn-seq}. The key difference lies in the statement and proof of the
analogue of \cref{prop:sn-seq:interrupt-handlers-sn}. As before, the
reducibility predicates and the normalisation proof only use skeletal types.

\begin{proposition}
  \label{prop:sn-reinstallable:interrupt-handlers-sn}
  Consider a continuation $\Gamma \types K : Y \multimap Z$, an operation $\op
  \in \Sigma$, and computations $\Gamma, x \of A_\op \types M_\op :
  \tysum{\typromise{X}}{\tyunit}$ and $\Gamma, p \of \typromise{X} \types N :
  Y$. If $K \circ (p)N \in \KRed_{\typromise{X}}^\Gamma$, $K[\mathit{wk}] \capp N \in
  \SN_{\!Z}$, and $M_\op[\mathit{lift}(\rho)][V/x] \in
  \CRed_{\tysum{\typromise{X}}{\tyunit}}^{\Gamma'}$, for all $\Gamma'$, $\rho
  \in \Ren{\Gamma}{\Gamma'}$, and $V \in \VRed_{A_\op}^{\Gamma'}$, then we have that
  $K \capp \tmwith{op}{x}{M_\op}{p}{N} \in \SN_{\!Z}$.
\end{proposition}

\begin{proof}
  The proof proceeds by induction on the lexicographic order of three measures:
  $|K|_\downarrow$, $|K| + \mathit{max}(K\capp N)$, and the structure of $N$, where
  $|K|_\downarrow$ denotes the number of interrupts $\copin{op}{V}$ in $K$, and
  we show that all reducts of $K \capp \tmwith{op}{x}{M_\op}{p}{N}$ are strongly
  normalising. Compared to \cref{prop:sn-seq:interrupt-handlers-sn}, we now also
  need to count the interrupts in $K$, via $|K|_\downarrow$, to affirm that
  during execution interrupts move inwards, and no new interrupts appear in
  computations---new interrupts only appear in parallel processes via the
  broadcast rules. For more details, we refer the reader to
  \cref{appendix:prop:sn-reinstallable:interrupt-handlers-sn}, or to our Agda
  formalisation.  
\end{proof}

After also adapting other propositions concerning interrupt handlers to
reinstalling, we can prove the overall reducibility result, from which it
follows that the sequential part of this extension of \lambdaAEff~is strongly
normalising, in both the skeletal and effect-typed versions.

\begin{theorem}
  \label{thm:sn-reinstallable:simulation}
  If $\Gamma \types M : \tycomp{X}{(\o,\i)}$, then $| \Gamma | \types M : | X |$, and 
  thus $M$ is strongly normalising.
\end{theorem}

%% file: normalisation-parallel.tex
\section{Strong Normalisation for Parallel Processes}
\label{sec:sn-parallel}

We conclude our normalisation results by showing that without adding
reinstallable interrupt handlers, the \emph{entire} \lambdaAEff-calculus presented in
\cref{sec:aeff} is in fact strongly normalising.

\subsection{Parallel Processes and Reinstallable Interrupt Handlers}

We first recall from the work of Ahman and Pretnar~\cite{Ahman:LMCS} that in the
extension of \lambdaAEff~with reinstallable interrupt handlers, the parallel
processes are not strongly normalising. The same is true for our variant of
reinstallable interrupt handlers. Consider the processes
\[
  \begin{array}{@{}l}
    ~~~\tmrun{\big(\tmopoutbig{ping}{\tmunit}{\tmwithbig{pong}{x}{\tmopout{ping}{\tmunit}{\tmkw{reinstall}}}{p}{\tmreturn{p}}}\big)}
    \\
    \tmkw{||}
    \\
    ~~~\tmrun{\big(\tmopoutbig{pong}{\tmunit}{\tmwithbig{ping}{x}{\tmopout{pong}{\tmunit}{\tmkw{reinstall}}}{p}{\tmreturn{p}}}\big)}.
  \end{array}
\]
These two parallel processes send each other an infinite number of ping-pong
signals. This means that by proving below that in \lambdaAEff~without
reinstallable interrupt handlers (and without general recursion) the parallel
processes are strongly normalising, we demonstrate that the addition of
reinstallable interrupt handlers adds real expressive and computational power.

\subsection{Parallel Processes Without Reinstallable Interrupt Handlers}
\label{sec:sn-parallel:tree-shaped}

As illustrated by the previous example, if we want to prove the strong
normalisation of parallel processes, it is not enough if all we know is that the
corresponding sequential computations are strongly normalising. While
computations including reinstallable interrupt handlers are strongly
normalising, the parallel processes containing them are not. So we need some
additional information about computations to prove normalisation for parallel
processes.

A convenient distinguishing feature between computations for which parallel
processes are strongly normalising, and for which they are not, turns out to be
effect-typing. Namely, the normalisation proof below relies on considering
computations effect-typed using \emph{finite} effect annotations $\i \in I$,
given by defining $I$ inductively, as the least fixed point of $\Phi (Z) \defeq
\sig \Rightarrow (O \times Z)_\bot$ from \cref{sec:aeff:seq}. This rules out
examples involving reinstallable interrupt handlers because to type such
examples we would need a coinductive definition of $I$ (see~\cite{Ahman:LMCS}).

Beyond also making use of effect-typing, this section differs from the
previous sections in that we do not need to employ a Girard-Tait-style
reducibility argument. This is largely because \lambdaAEff~does not contain
first-class processes. Instead, the normalisation proof can proceed by direct
induction on the lexicographic order of four measures that we define next.

\emph{Size of effect annotations:} 
For a finite $\i$, denote with $\sizei{\i}$ the \textit{size} of
$\i$: the number of internal nodes when $\i$ is considered as a tree. We then
define $\sizei{(\o,\i)}$ as $\sizei{\i}$. In fact, only the size of $\i$ counts below, and the $\o$-annotations could be omitted altogether.
We write $\sizei{P}$ for the sum of the sizes $\sizei{\i}$ over all
individual computations in the $\tmkw{run}$-leaves of a process $P$. 

\begin{proposition}
  \label{prop:sn-par:act-size-decr}
  If $\op\not\in\i$, then $\sizei{\opincomp{op}{(\o,\i)}} = \sizei{\i}$, and if $\op\in\i$, then $\sizei{\opincomp{op}{(\o,\i)}} < \sizei{\i}$.
\end{proposition}

\emph{Maximum number of outgoing signals:} By \cref{thm:sn-seq:sn}, every
\lambdaAEff-computation is strongly normalising. Since a \lambdaAEff-computation has only finitely many
one-step reducts, it has finitely many normal forms.
We note that when construed as a tree, a normal form of a
\lambdaAEff-computation has a shape in which all outgoing signals are towards
the root, yet-to-be-triggered interrupt handlers are in the intermediate nodes,
and the leaves contain $\tmkw{return}$s and blocked $\tmkw{await}$s
(see~\cite{Ahman:LMCS} for more details). 
Below we write $\mathit{max}\tmout(M)$ for the maximum number of top-level signals in a
normal form of $M$, i.e., $\mathit{max}\tmout(M)$ is the maximum number of signals $M$
can emit while it reduces.
Further, we write $\mathit{max}\tmout(P)$ for the sum of $\mathit{max}\tmout(M)$ over all
individual computations $M$ in the $\tmkw{run}$-leaves of a process $P$.

Note that if a computation does not have a handler for an interrupt, then acting with the
corresponding interrupt cannot reveal any new signals in this computation:
\begin{proposition}
  \label{prop:sn-par:max-signals-eq} 
  If $\Gamma \types M : \tycomp{X}{(\o,\i)}$ and $\op \not\in \i$, then
  $\mathit{max}\tmout(\tmopin{op}{V}{M}) = \mathit{max}\tmout(M)$.
\end{proposition}

\emph{Parallel shapes:} We define a notion of \emph{parallel shapes} $S, T, \ldots$ to model the parallel part of a process that is independent of the individual computations in the $\tmkw{run}$-leaves:
\[
  S, T \bnfis \tmkw{run} \bnfor \typar{S}{T} \bnfor \tmin S \bnfor \tmout S
\]
We can define a reduction relation $S \reduces T$ for parallel shapes by
restricting \cref{fig:par-reduction-rules} to the grammar of parallel shapes,
e.g., $\tmin{\tmkw{run}} \reduces \tmkw{run}$ and $\tmpar{(\tmout S)}{T}
\reduces \tmout{(\tmpar{S}{(\tmin{T})})}$.
It is easily shown that this reduction relation is strongly normalising. For a
shape $S$, we write $\mathit{max}_{sh}(S)$ for the length of the longest reduction sequence
starting from $S$.
We also note that every parallel process $P$ has a straightforwardly determined
parallel shape, which we write as $\mathit{shape}(P)$.

\emph{Maximum number of $\tmkw{run}$-reduction steps:} For the last induction measure, we define $\mathit{max}_{\tmkw{run}}(P)$ as the sum of $\mathit{max}(M)$ for all individual computations $M$ in the $\tmkw{run}$-leaves of $P$.

We can now prove that the parallel part of \lambdaAEff~from \cref{sec:aeff} is
also strongly normalising.

\begin{theorem}\label{thm:sn-par:sn}
  Every well-typed parallel process $\Gamma \types P : C$ is strongly normalising.
\end{theorem}
\begin{proof}
By induction on the lexicographic order of $\sizei{P}$,
$\mathit{max}\tmout(P)$, $\mathit{max}_{sh}(\mathit{shape}(P))$, and $\mathit{max}_{\tmkw{run}}(P)$, we show
that all reducts of the process $P$ are strongly normalising.

The rules \eqref{r16}, \eqref{r17}, \eqref{r19}, and \eqref{r20} reduce only the parallel shape of $P$, so they decrease $\mathit{max}_{sh}(\mathit{shape}(P))$. The rules \eqref{r14}--\eqref{r15} decrease respectively $\mathit{max}_{\tmkw{run}}(P)$ and $\mathit{max}\tmout(P)$.

It remains to consider the rule \eqref{r18}, i.e., the propagation of interrupts
into one of the individual computations as $\tmopin{op}{V}{\tmrun M} \reduces
\tmrun {(\tmopin{op}{V}{M})}$. Suppose the resulting computation is typed as
$\Gamma \types \tmopin{op}{V}{M} : \tycomp{X}{\opincomp{op}{(\o,\i)}}$, where
${\Gamma \types M : \tycomp{X}{(\o,\i)}}$. If $\op\in\i$, then by
\cref{prop:sn-par:act-size-decr} we have $\sizei{\opincomp{op}{(\o,\i)}} <
\sizei{\i}$, so this decreases the first induction measure. If
$\op\not\in\i$, then by \cref{prop:sn-par:max-signals-eq} we have
$\mathit{max}\tmout(\tmopin{op}{V}{M}) = \mathit{max}\tmout(M)$, and by
\cref{prop:sn-par:act-size-decr}, $\sizei{\opincomp{op}{(\o,\i)}} =
\sizei{\i}$, so the first two induction measures remain unchanged. In that
case, $\mathit{max}_{sh}(\mathit{shape}(P))$ is decreased by the parallel shape reduction rule
$\tmin{\tmkw{run}} \reduces \tmkw{run}$.
\end{proof}


\subsection{Strong Normalisation for a Flat Model of Parallel Processes}
\label{sec:sn-par:flat-processes}

\newcommand{\flatpar}{\mathrel{\tmkw{||}}}

We conclude this section by discussing an alternative presentation of parallel
processes, as \emph{flat non-empty lists of individual computations}, written
$P,Q \bnfis M_1 \flatpar M_2 \flatpar \ldots \flatpar M_n$, and also establish
strong normalisation for such processes. This presentation is of interest
because instead of the algebraically more natural tree-shaped presentation of
processes discussed so far, the flat list presentation is what one might
consider more natural for implementation.

We write $\Gamma \types P$ for the (skeletal) typing of such processes, 
with the typing rules requiring that each
of the individual computations $M_i$ in $P$ is well-typed at some computation type.

We define the small-step operational semantics of such processes using two relations:
\begin{mathpar}
  \coopinfer{}{
    M_i \reduces M_i'
  }{
    M_1 \flatpar \ldots \flatpar M_i \flatpar \ldots \flatpar M_n \reduces M_1 \flatpar \ldots \flatpar M_i' \flatpar \ldots \flatpar M_n
  }
  \and
  \coopinfer{}{
    P \overset{\copout{op}{V}}{\reduces} Q
  }{
    P \reduces Q
  }
\end{mathpar}
and
\begin{mathpar}
  \coopinfer{}{
  }{
    (\tmopout{op}{V}{M_1}) \flatpar M_2 \flatpar \ldots \flatpar M_n 
     \overset{\copout{op}{V}}{\reduces} 
     M_1 \flatpar \tmopin{op}{V}{M_2} \flatpar \ldots \flatpar \tmopin{op}{V}{M_n}
  }
  \and
  \coopinfer{}{
    M_2 \flatpar \ldots \flatpar M_n
    \overset{\copout{op}{V}}{\reduces}
    M'_2 \flatpar \ldots \flatpar M'_n
  }{
    M_1 \flatpar M_2 \flatpar \ldots \flatpar M_n 
    \overset{\copout{op}{V}}{\reduces} 
    \tmopin{op}{V}{M_1} \flatpar M'_2 \flatpar \ldots \flatpar M'_n
  }
\end{mathpar}

The first reduction relation $P \reduces Q$ expresses how processes reduce at
the top level: either one of the individual computations makes a step, or one of
the individual computations issues a signal $\op$ with payload $V$ that gets
propagated to other processes as an interrupt.

The second reduction relation $P \overset{\copout{op}{V}}{\reduces} Q$ then
formalises how these signals are propagated to other processes as interrupts: if
the signal is issued in the leftmost process, then the first rule propagates it
immediately to all the processes to the right of it, while nested repeated
applications of the second rule also propagate the signal to the processes to the left
of it. The second rule is due to the formal inductive definition of non-empty lists of processes. In a more informal presentation, we could replace the two rules with a single rule of the form:
\begin{mathpar}
  \coopinfer{}{
  }{
    M_1 \flatpar \ldots \flatpar \tmopout{op}{V}{M_i} \flatpar \ldots \flatpar M_n 
     \overset{\copout{op}{V}}{\reduces} 
     \tmopin{op}{V}{M_1} \flatpar \ldots \flatpar M_i \flatpar \ldots \flatpar \tmopin{op}{V}{M_n}
  }
\end{mathpar}

Similarly to the normalisation of the tree-shaped model of parallel processes in
\cref{sec:sn-parallel:tree-shaped}, in the absence of general recursion and
without adding reinstallable interrupt handlers, and when the individual
computations $M_i$ are effect-typed using finite effect annotations $\i$ (as in
\cref{sec:sn-parallel:tree-shaped}), we can prove that this variant of parallel
processes is also strongly normalising. 

\begin{theorem}
  Every (skeletally) well-typed parallel process $\Gamma \types P$ is strongly
  normalising.
\end{theorem}

\begin{proof}
  The proof is similar to the proof of \cref{thm:sn-par:sn}. Here we proceed by
  induction on the lexicographic order of three measures: the sum $\sizei{P}$ of
  the sizes of the effect annotations of the individual computations in $P$, the
  sum $\mathit{max}\tmout(P)$ of the maximum number of outgoing signals in the individual
  computations in $P$, and the sum $\mathit{max}(P)$ of the maximum number of reduction
  steps of the individual computations in $P$, and we show that all the reducts
  of $P$ are strongly normalising. Further details can be found in our Agda
  formalisation.
\end{proof}

We leave formally relating this flat-lists-based model of parallel processes to
the original tree-shaped model of parallel processes discussed in
\cref{sec:aeff:par,sec:sn-parallel:tree-shaped} for future work.

%% file: conclusion.tex
\section{Conclusion}
\label{sec:conclusion}

We have demonstrated that the Girard-Tait-style reducibility approach to proving
strong normalisation, and its $\top\top$-lifting extension to effectful
languages by Lindley and Stark, are well-suited for proving normalisation
properties of \lambdaAEff, a core calculus for asynchronous programming with
algebraic effects~\cite{Ahman:POPL,Ahman:LMCS}. On the one hand, studying the
normalisation properties of this calculus is desirable because it allows one to
model many compelling and natural examples, such as pre-emptive multi-threading,
(cancellable) remote function calls, multi-party applications, and a parallel
variant of runners of algebraic effects. On the other hand, studying the
normalisation properties of \lambdaAEff~is also interesting because it
contains a number of advanced and challenging features. Our results are also
formalised in Agda.

In future work, we plan to extend our results to the other extensions that Ahman
and Pretnar proposed for \lambdaAEff~\cite{Ahman:LMCS}: modal types that allow
signals and interrupts to carry higher-order data, a stateful variant of
reinstallable interrupt handlers, and the dynamic spawning of new processes.
Moreover, observe that as the use of the sum type in the reinstallable interrupt
handlers of \cref{sec:sn-reinstallable:sn} somewhat resembles Elgot
iteration~\cite{bloom-ezik:iteration-theories}, it would be interesting to study
a formal connection between the two. In a related direction, we also want to
investigate how to present our strongly normalising reinstallable interrupt
handlers without making explicit use of sum types, closer to the original
abstract style of Ahman and Pretnar.

%% file: appendix.tex

\section{Proof of \cref{prop:sn-seq:interrupt-handlers-sn}}
\label[appendix]{appendix:prop:sn-seq:interrupt-handlers-sn}

First, we state a proposition saying that removing a signal from inside a
computation does not increase the maximum number of reduction steps starting
from it. For better readability, below we omit explicit actions of (weakening)
renamings, and some type and context indices.

\begin{proposition}
  \label{prop:max-steps-infix-signal} 
  Suppose $K \capp \big(\tmopout{op}{V}{N}\big) \in \SN$. Then the following
  statements hold:
  \begin{itemize}
    \item $K \capp N \in \SN$ and
    \item $\mathit{max}(K \capp N) \leq \mathit{max}\big(K \capp (\tmopout{op}{V}{N})\big).$
  \end{itemize} 
\end{proposition}

\begin{proof}
  The strongly normalising computation $K \capp \big(\tmopout{op}{V}{N}\big)$ reduces, in $|K|$ steps, to $\tmopout{op}{V}{K \capp N}$, which is thus also strongly normalising. It is then clear that $K \capp N \in \SN$.

  For any reduction sequence $r$ starting from $K \capp N$, there is a reduction
  sequence, not shorter than $r$, starting from $K \capp
  \big(\tmopout{op}{V}{N}\big)$: first reduce $K \capp
  \big(\tmopout{op}{V}{N}\big)$ to $\tmopout{op}{V}{K \capp N}$ and then perform
  the sequence $r$ under the signal. This proves the desired inequality.
\end{proof}

\paragraph*{Proof of \cref{prop:sn-seq:interrupt-handlers-sn}:}

\newcommand{\refreduces}[1]{\overset{\eqref{#1}}{\reduces}}

The proof proceeds by induction on the lexicographic order of $|K| + \mathit{max}(K \capp
N)$ and the structure of $N$. We show that all reducts of $K \capp
\tmwith{op}{x}{M_{\op}}{p}{N}$ are strongly normalising. The rules
\eqref{r4} and \eqref{r10} are easy: they decrease $|K|$ and leave $K\capp N$
unchanged, so we omit these cases. The remaining applicable rules are \eqref{r9},
\eqref{r6}, and \eqref{r13}.

We first consider the reduction rule \eqref{r9}:
\[
  \begin{array}{l@{~~} l}
    & K \capp \tmwith{op}{x}{M_\op}{p}{N}
    \\[0.5ex]
    \refreduces{r9} & K' \capp \big(\tmlet{p}{M_\op[V/x]}{\tmopin{op}{V}{N}}\big), 
  \end{array}
\]
where $K = K'\circ\copin{op}{V}$ and $p\not\in fv(V)$.
Observe that
\[
  K' \capp \big(\tmlet{p}{M_\op[V/x]}{\tmopin{op}{V}{N}}\big) = \left(K'\circ(p)(\tmopin{op}{V}{N})\right) \capp M_\op[V/x].
\]
We note that the term on the right is strongly normalising, because by assumption we have $M_\op[V/x] \in \CRed_{\typromise X}$, and by \cref{prop:sn-seq:interrupt-handlers-auxiliary-results} it holds that $K'\circ(p)(\tmopin{op}{V}{N}) \in \KRed_{\typromise X}$.

Next, we consider the reduction rule \eqref{r6}:
\[
  \begin{array}{l@{~~} l}
    & K \capp \tmwith{op}{x}{M_\op}{p}{N}
    \\[0.5ex]
    \refreduces{r6} & K \capp \left(\tmopout{op'}{V}{\tmwith{op}{x}{M_\op}{p}{N'}}\right), 
  \end{array}
\]
where $N = \tmopout{op'}{V}{N'}$ and $p\not\in fv(V)$.
By \cref{prop:sn-seq:signals-sn}, it suffices to show that
\[
  K \capp \tmwith{op}{x}{M_\op}{p}{N'} \in \SN.
\]
Here we want to apply the induction hypothesis, because $N'$ is structurally smaller than $N = \tmopout{op'}{V}{N'}$, but we need to be careful not to increase the first induction measure. But $|K|$ stays the same, and by \cref{prop:max-steps-infix-signal} we have that $\mathit{max}(K \capp N') \leq \mathit{max}(K \capp N)$.

Finally, we consider the reduction rule \eqref{r13}, i.e., the evaluation context rule:
\[
  \begin{array}{l@{~~} l}
    & K \capp \tmwith{op}{x}{M_\op}{p}{N}
    \\[0.5ex]
    \refreduces{r13} & K \capp \tmwith{op}{x}{M_\op}{p}{N'}, 
  \end{array}
\]
where it is assumed that $N \reduces N'$.
Then we have also that $K \capp N \reduces K \capp N'$, so the first induction
measure is decreased, and by
\cref{prop:sn-seq:interrupt-handlers-auxiliary-results} we have that $K \circ
(p)N' \in \KRed_{\typromise X}$. Therefore, all the conditions for applying the
induction hypothesis are satisfied.
\qed


\section{Proof of \cref{prop:sn-reinstallable:interrupt-handlers-sn}}
\label[appendix]{appendix:prop:sn-reinstallable:interrupt-handlers-sn}

First, we state two additional auxiliary results about the interaction
of strong normalisation with the application of continuations $K$ to 
computations directly in redex form.

\begin{proposition}\label{prop:let-return}
  If $K\capp N[V/x]\in\SN$, then $K\capp(\tmlet{x}{\tmreturn{V}}{N})\in\SN$.
\end{proposition}

\begin{proof}
  We observe that the computation $K\capp(\tmlet{x}{\tmreturn{V}}{N})$ reduces
  to only $K\capp N[V/x]$, which is strongly normalising by assumption.
\end{proof}

\begin{proposition}\label{prop:match-inject}
  \mbox{}
  \begin{itemize}
    \item If $K\capp M[V/x]\in\SN$, then $K\capp(\tmmatch{(\tminl{V})}{\tminl{x}\mapsto M,\tminr{y}\mapsto N}{})\in\SN$.
    \item If $K\capp N[W/y]\in\SN$, then $K\capp(\tmmatch{(\tminr{W})}{\tminl{x}\mapsto M,\tminr{y}\mapsto N}{})\in\SN$.
  \end{itemize}
\end{proposition}

\begin{proof}
  We observe that $K\capp(\tmmatch{(\tminl{V})}{\tminl{x}\mapsto
  M,\tminr{y}{}\mapsto N}{})$ reduces to only $K\capp M[V/x]$, which is strongly
  normalising by assumption, and analogously for the other statement about
  pattern-matching on a right injection.
\end{proof}

\paragraph*{Proof of \cref{prop:sn-reinstallable:interrupt-handlers-sn}:}

The proof proceeds by induction on the lexicographic order of three measures:
$|K|_\downarrow$, $|K| + \mathit{max}(K\capp N)$, and the structure of $N$, where
$|K|_\downarrow$ denotes the number of interrupts $\copin{op}{V}$ in $K$, and we
show that all reducts of the computation $K \capp
\tmwith{op}{x}{M_\op}{p}{N}$ are strongly normalising. 
For better readability, we omit explicit actions of (weakening) renamings and
context indices from the reducibility predicates.

Below we discuss the case for the most interesting reduction rule \eqref{r9} in detail:
\[
  \begin{array}{l@{~~} l}
    & K \capp \tmwith{op}{x}{M_\op}{p}{N}
    \\[0.5ex]
    \refreduces{r9} & K' \capp \big(\tmlet{p}{(\tmlet{y}{M_\op[V/x]}{\tmmatch{y}{\tminl{z}{}\mapsto \tmreturn{z},\tminr{w}{}\mapsto R}{}})\\[0.5ex]&\hspace{0.9cm}}{\tmopin{op}{V}{N}}\big)
  \end{array}
\]
where $K = K'\circ\copin{op}{V}$ and $p\not\in fv(V)$, and $R = \tmwith{op}{x}{M_\op}{q}{\tmreturn{q}}$.

We first notice that 
\[
  \begin{array}{l@{~~} l}
    & K' \capp \big(\tmlet{p}{(\tmlet{y}{M_\op[V/x]}{\tmmatch{y}{\tminl{z}{}\mapsto \tmreturn{z},\tminr{w}{}\mapsto R}{}})\\[0.5ex]&\hspace{0.9cm}}{\tmopin{op}{V}{N}}\big)
    \\[1ex]
    = & \big(K'\circ(p)(\tmopin{op}{V}{N}) \circ (y) (\tmmatch{y}{\tminl{z}\mapsto \tmreturn{z},\tminr{w}{}\mapsto R})\big) \capp M_\op[V/x]
  \end{array}
\]
Next, by assumption, we have that $M_\op[V/x]\in\CRed_{\tysum{\typromise{X}}{\tyunit}}$, so it suffices to show
\[
  K' \circ (p)(\tmopin{op}{V}{N}) \circ (y)(\tmmatch{y}{\tminl{z} \mapsto
  \tmreturn{z},\tminr{w}\mapsto R}{}) \in \KRed_{\tysum{\typromise{X}}{\tyunit}}.
\]
Now, by the definition of $\KRed$ and \cref{prop:let-return}, it suffices that for all $V'\in\VRed_{\tysum{\typromise{X}}{\tyunit}}$,
\[
  \big(K' \circ (p)(\tmopin{op}{V}{N})\big) \capp \tmmatch{V'}{\tminl{z} \mapsto \tmreturn{z},\tminr{w} \mapsto R}{} \in \SN, 
\]
which by the definition of $\SRed$ is equivalent to having to prove that
\[
  K' \circ (p)(\tmopin{op}{V}{N}) \circ ((z)(\tmreturn{z}), (w)R) \in \SRed_{\typromise{X},\tyunit}.
\]
Next, we observe that by the definition of $\SRed$ and \cref{prop:match-inject},
this is equivalent to having to prove the following two statements:
\begin{itemize}
  \smallskip
  \item $(K'\circ (p)\tmopin{op}{V}{N})\capp(\tmreturn{W})\in\SN, \text{for all } W\in\VRed_{\typromise{X}}$, and 
  \smallskip
  \item $(K'\circ (p)\tmopin{op}{V}{N})\capp (\tmwith{op}{x}{M_\op}{q}{\tmreturn{q}})\in\SN$.
\end{itemize}

Note that the first of these statements is exactly the definition of $K' \circ (p)(\tmopin{op}{V}{N}) \in \KRed_{\typromise{X}}$, which holds by \cref{prop:sn-seq:interrupt-handlers-auxiliary-results} and the hypothesis that $K \circ (p)N \in \KRed_{\typromise{X}}$.

For proving the second of the above statements, we apply the induction hypothesis.
For this, we begin by observing that the first induction measure has decreased:
\[
  |K' \circ (p)(\tmopin{op}{V}{N}) |_\downarrow = |K'|_\downarrow < |K'\circ\copin{op}{V}|_\downarrow = |K|_\downarrow.
\]
Secondly, we observe that $K' \circ (p)(\tmopin{op}{V}{N}) \circ (q)(\tmreturn{q})\in\KRed_{\typromise{X}}$, because it is equivalent to the statement $K' \circ (p)(\tmopin{op}{V}{N}) \in \KRed_{\typromise{X}}$, which holds by the above.
Finally, we note that $\big(K' \circ (p)(\tmopin{op}{V}{N})\big) \capp (\tmreturn{q})\in\SN$, again because $K'\circ (p)(\tmopin{op}{V}{N}) \in \KRed_{\typromise{X}}$.
Therefore, all the conditions to be able to apply the induction hypothesis are
satisfied.
This concludes the proof for the reduction rule \eqref{r9}.
\qed


\section{Proof of \cref{prop:sn-par:act-size-decr}}
\label[appendix]{appendix:prop:sn-par:act-size-decr}

We can view a finite effect annotation $\i \in I$ as a finite set
$paths(\i)$ of paths from the root to the internal nodes of $\i$: a \emph{path} is a tuple $(\tm{op}_1, \dotsc, \tm{op}_n)$ of operations.
We write $paths(\o,\i)$ for $paths(\i)$.
Next, we note that $\lvert paths(\i) \rvert = \sizei{\i}$, and that
$\op$ acts on an $\i$ path-by-path, as
\[
  paths(\opincomp{op}{(\o,\i)}) = \{ \opincomp{op}{p} \mid p \in paths(\i) \},
\]
where the action $\opincomp{op}{p}$ of an interrupt $\op$ on a path $p$ is
defined as follows:
\[
\opincomp{op}{p} =
  \begin{cases}
    p' & \mbox{if } p = (\tm{op}, p') \\
    p & \mbox{otherwise, i.e., if $p$ is empty or its first operation is not $\tm{op}$}
  \end{cases}
\]
Now, to prove the proposition, we consider the following two cases based on
whether $\op \in \i$.

If $\op \not\in \i$, then there is no path in $paths(\i)$ which starts
with $\op$, so in that case $\opincomp{op}{(\o,\i)} = (\o,\i)$, and therefore
$\sizei{\opincomp{op}{(\o,\i)}} = \sizei{\i}$.

If $\op \in \i$, then it holds that the length-one path $(\op) \in
paths(\i)$ and that the empty path $() \in paths(\i)$ (the latter holds if
and only if $\i \neq \emptyset$). Acting with $\op$ sends both $(\op)$ and $()$
to the same element $()$. From this, the desired strict inequality follows. More
formally,
\[
  \begin{array}{l@{~~} l}
    & \sizei{\opincomp{op}{(\o,\i)}} \\[0.5ex]
    = & \lvert paths(\opincomp{op}{(\o,\i)})\rvert \\[0.5ex]
    = & \lvert \{ \opincomp{op}{p} \mid p \in paths(\i) \} \rvert \\[0.5ex]
    = & \lvert \{ \opincomp{op}{p} \mid p \in paths(\i)\setminus\{() \} \} \rvert \\[0.5ex]
    \leq & \lvert paths(\i) \setminus \{()\} \rvert \\[0.5ex]
    < & \lvert paths(\i) \rvert \\[0.5ex]
    = & \sizei{\i}.
  \end{array}
\]
\qed